\documentclass[a4paper, 12pt, twoside]{article}

\usepackage[a4paper, margin=1in]{geometry}
\usepackage[utf8]{inputenc}
\usepackage[T1]{fontenc}
\usepackage[english]{babel}
\usepackage{textcomp} 
\usepackage{graphicx}
\DeclareGraphicsExtensions{.jpg,.mps,.pdf,.png,.gif}
\usepackage[french]{varioref} 
\usepackage[]{amsmath,amssymb,amsfonts}

\usepackage{float}
\usepackage{bbm}
\usepackage{amsthm}
\usepackage{dsfont}
\usepackage{color}
\usepackage{natbib}
\usepackage{bm}
\usepackage{subcaption}

\usepackage[hyperindex,breaklinks]{hyperref}

\newcommand{\Ben}{\begin{enumerate}}
\newcommand{\Een}{\end{enumerate}}
\newcommand{\Bit}{\begin{itemize}}
\newcommand{\Eit}{\end{itemize}}
\newcommand{\Beq}{\begin{equation}}
\newcommand{\Eeq}{\end{equation}}
\newcommand{\Ba}{\begin{align*}}
\newcommand{\Ea}{\end{align*}}
\newcommand{\Mb}{\mathbf}
\newcommand{\Mr}{\mathrm}

\newcommand{\Mbb}{\mathbb}
\newcommand{\Tb}{\textcolor{black}}

\newtheorem{Th}{Theorem}
\newtheorem{Lem}{Lemma}
\newtheorem{Prop}{Proposition}

\newtheorem{Rq}{Remark}

\title{Correlation of powers of Hüsler--Reiss vectors and Brown--Resnick fields, and application to insured wind losses}
\date{March 1, 2022}
\begin{document}
\author{
Erwan Koch\footnote{EPFL (Institute of Mathematics): EPFL SB MATH MATH-GE, MA A1 354 (B\^atiment MA), Station 8, 1015 Lausanne, Switzerland. \newline Email: erwan.koch@epfl.ch}
}

\maketitle

\begin{abstract}
Hüsler--Reiss vectors and Brown--Resnick fields are popular models in multivariate and spatial extreme-value theory, respectively, and are widely used in applications. We provide analytical formulas for the correlation between powers of the components of the bivariate Hüsler--Reiss vector, extend these to the case of the Brown--Resnick field, and thoroughly study the properties of the resulting dependence measure. The use of  correlation is justified by spatial risk theory, while power transforms are insightful when taking correlation as dependence measure, and are moreover very suited damage functions for weather events such as wind extremes or floods. This makes our theoretical results worthwhile for, e.g., actuarial applications. We finally perform a case study involving insured losses from extreme wind speeds in Germany, and obtain valuable conclusions for the insurance industry.
 
\medskip

\end{abstract}

\section{Introduction}
\label{Sec_Intro}

Extreme-value theory offers many statistical techniques and models useful in various fields such as finance, insurance and environmental sciences. Max-stable random vectors \citep[e.g.,][]{deHaanResnick1977Limit} naturally arise when extending univariate extreme-value theory to the multidimensional setting, and several parametric multivariate max-stable distributions, such as the Hüsler--Reiss model \citep{husler1989maxima}, have been proposed. Max-stable random fields \citep[e.g.,][]{haan1984spectral, de2007extreme, davison2012statistical} constitute an infinite-dimensional generalization and are particularly suitable to model the temporal maxima of a given variable at all points in space since they represent the only possible non-degenerate limiting field of pointwise maxima taken over suitably rescaled independent copies of a field \citep[e.g.,][]{haan1984spectral}. One famous example is the Brown--Resnick field \citep{brown1977extreme, kabluchko2009stationary} which, owing to its flexibility, is generally a good model for spatial extremes of environmental variables. Finite-dimensional distributions of the Brown--Resnick field are Hüsler--Reiss distributions so there is a natural and close link between Hüsler--Reiss vectors and Brown--Resnick fields.

Our main theoretical contributions are explicit formulas for the correlation between powers of the components of bivariate Hüsler--Reiss random vectors, analytical expressions of the spatial correlation function of powers of Brown--Resnick fields, and a careful study of its properties; some results are rather technical to obtain. Studying the correlation function of a field is prominent as it naturally appears when computing the variance of the spatial integral of that field \citep[e.g.,][]{koch2019SpatialRiskAxioms}. If the field models an insured cost, its spatial integral models the total insured loss over the integration region, and its variance is thus of interest for any insurance company. The correlation function also explicitly shows up in the standard deviation of the central limit theorem of the field, and is thus key for the behaviour of the spatial integral when the size of the integration region becomes large \citep[e.g.,][]{koch2019SpatialRiskAxioms}. Moreover, despite its drawbacks,
correlation is commonly used in the finance/insurance industry, making its study useful from a practical viewpoint. Finally, the criticism that it does not properly capture extremal dependence is somewhat irrelevant here as we consider the correlation between random variables which already model extreme events.

It is often insightful to consider the correlation between various powers of two random variables rather than focusing only on the correlation between these variables. First, applying simple non-linear transformations such as the absolute value or powers before taking the correlation sometimes allows one to detect and characterize a strong dependence that would not have been spotted using otherwise; this partially alleviates the defect that correlation only captures linear dependence. In finance, it is common to look at the autocorrelation of powers of the absolute values of asset returns. Returns generally do not exhibit any significant autocorrelation \citep[e.g.,][]{cont2001empirical} whereas their squares or other power values \citep[see, e.g.,][who consider powers ranging from $0.125$ to $5$]{ding1993long} show a significantly positive serial correlation. Second, taking powers may be useful for estimation.
Let $X_1, X_2$ be random variables whose joint distribution function depends on a parameter $\bm{\theta}$. If an explicit formula is available for the correlation between $X_1^{\beta}$ and $X_2^{\beta}$, where $\beta$ belongs to an appropriate space, then one can think of estimating $\bm{\theta}$ by equating that expression with its empirical counterpart, and searching for the value of $\beta$ leading to the optimal corresponding estimator. Such an approach may be notably useful for max-stable random fields, for which estimation is arduous.

Considering powers of random variables is also valuable when these variables are used to model the impact of natural disasters such as, e.g., windstorms or floods. According to physics, the total cost arising from damaging wind to a specific structure should increase as the square \citep[e.g.,][]{simiu1996wind} or the cube \citep[e.g.,][]{lamb1991historic, emanuel2005increasing, powell2007tropical} of the maximum wind speed. Moreover, several studies exploring insured costs have found that power-laws with much higher exponents are appropriate \citep[e.g.,][]{prahl2012applying}. A commonly used damage function for flood is $D(z)=z/(z+1)$, where $z>0$ is the water level measured in meters \citep[e.g.,][]{hinkel2014coastal, prahl2016damage} and so the destruction percentage approximately follows a power-law with exponent unity for levels much below one meter. Thus, as max-stable vectors and fields are suited to model componentwise and pointwise maxima, studying their powers is worthwhile for assessing costs from extreme wind or flood events.

In the second part of the paper, we use our theoretical results to study the spatial dependence of insured losses from extreme wind speed for residential buildings over a large part of Germany. We use ERA5 (European Centre for Medium-Range Weather Forecasts Reanalysis 5th Generation) wind speed reanalysis data on 1979--2020 to derive seasonal pointwise maxima, we fit the Brown--Resnick and Smith random fields, and use the appropriate power damage function for the considered region, according to \cite{prahl2012applying}. The best fitted model leads to a correlation displaying a slow decrease with the distance. We also consider other power values and we find that, for a fixed distance, the correlation between insured costs evolves only slightly with the value of the damage power; this is useful information for insurance companies.

The rest of the paper is organized as follows. Section \ref{Sec_TheoreticalResults} first briefly reviews Hüsler--Reiss vectors and Brown--Resnick fields, and then details our main theoretical contributions. We present our case study in Section \ref{Sec_Application},  and Section \ref{Sec_Conclusion} summarizes our main findings and provides some perspectives. All the proofs are gathered in the Appendix. The code and data required to reproduce the results of the case study will be available in a publication on the Zenodo repository. Note that some elements of this article are revised versions of results from Sections~2.2 and 3 and Appendix A of the unpublished work by \cite{koch2019spatialpowers}. Throughout the paper, $'$ designates transposition and $\mathbb{N}_*=\mathbb{N} \backslash \{ 0 \}$.  

\section{Theoretical results}
\label{Sec_TheoreticalResults}

\subsection{Preliminaries}

A random variable $Z$ has the standard Fréchet distribution if $\Mbb{P}(Z \leq z)=\exp(-1/z), z>0$. A random vector $\bm{Z}=(Z_1, Z_2)'$ having standard Fréchet marginals is said to follow the bivariate Hüsler--Reiss distribution \citep{husler1989maxima} with parameter $h \in [0, \infty]$ if 
\begin{align}
& \quad \ \Mbb{P}(Z_1 \leq z_1, Z_2 \leq z_2) \nonumber \\& = H(z_1, z_2; h) = \exp \left( -\frac{1}{z_2} \Phi \left( \frac{h}{2} - \frac{\log(z_2/z_1)}{h} \right) -\frac{1}{z_1} \Phi \left( \frac{h}{2} - \frac{\log(z_1/z_2)}{h} \right) \right), \quad z_1, z_2 > 0.
\label{Eq_HuslerReissDist}
\end{align}
This is a popular and flexible distribution for max-stable random vectors, and the parameter $h$ interpolates between complete dependence ($h=0$) and independence ($h = \infty$). The $i$-th component, $i=1, 2$, of any bivariate max-stable vector follows the generalized extreme-value (GEV) distribution with location, scale and shape parameters $\eta_i \in \mathbb{R}$, $\tau_i >0$ and $\xi_i \in \mathbb{R}$. If $\bm{X}=(X_1, X_2)'$ is max-stable with such GEV parameters, then
\Beq
\label{Eq_LinkMaxstbVect_SimpleMaxstabVect}
X_i = 
\left \{
\begin{array}{ll}
\eta_i-\tau_i/\xi_i + \tau_i Z_i^{\xi_i}/\xi_i, & \quad \xi_i \neq 0, \\
\eta_i + \tau_i \log Z_i, & \quad \xi_i = 0,
\end{array}
\right.
\Eeq
where $(Z_1, Z_2)'$ is a max-stable vector with standard Fréchet marginal distributions. 

In the following, a max-stable random field with standard Fr\'echet margins will be called simple. The class of Hüsler--Reiss distributions is tightly linked to the Brown--Resnick random field \citep{brown1977extreme, kabluchko2009stationary} which is a flexible and widely used max-stable model. It is very suited to model, e.g., extremes of environmental data \citep[e.g.,][Section 7.4, in the case of rainfall]{davison2012statistical} as it allows realistic realizations as well as independence when distance goes to infinity. If $\{ W(\bm{x}) \}_{\bm{x} \in \mathbb{R}^d}$ is a centred Gaussian random field with stationary increments and with semivariogram $\gamma_W$, then the Brown--Resnick random field associated with the semivariogram $\gamma_W$  is defined by
\Beq
\label{Eq_Spectral_Representation_Stochastic_Processes}
Z(\bm{x}) = \bigvee_{i=1}^{\infty} U_i Y_i(\bm{x}), \quad \bm{x} \in \mathbb{R}^d,
\Eeq
where the $(U_i)_{i \geq 1}$ are the points of a Poisson point process on $(0, \infty)$ with intensity function $u^{-2} \mathrm{d}u$ and the $Y_i, i\geq 1$, are independent replications of 
$$ Y(\bm{x})=\exp \left( W(\bm{x})-\mathrm{Var}(W(\bm{x}))/2 \right), \quad \bm{x} \in \Mbb{R}^d,$$
where $\mathrm{Var}$ denotes the variance. It is a stationary\footnote{Throughout the paper, stationarity refers to strict stationarity.} and simple max-stable field whose distribution only depends on the semivariogram \citep[][Theorem 2 and Proposition 11, respectively]{kabluchko2009stationary}. Its finite-dimensional distribution functions are Hüsler--Reiss distributions \citep[][Remark 24]{kabluchko2009stationary} and, in particular, for any $\bm{x}_1, \bm{x}_2 \in \Mbb{R}^d$,
\Beq
\label{EqBivDistFuncBRField}
\Mbb{P}(Z(\bm{x}_1) \leq z_1, Z(\bm{x}_2) \leq z_2) = H \left(z_1, z_2; \sqrt{2 \gamma_W(\bm{x}_2-\bm{x}_1)} \right), \quad z_1, z_2>0.
\Eeq
A commonly used semivariogram is 
\Beq
\label{Eq_Power_Variogram}
\gamma_W(\bm{x})= \left( \| \bm{x} \| / \kappa \right)^{\psi}, \quad \bm{x} \in \mathbb{R}^d,
\Eeq
where $\kappa>0$ and $\psi \in (0, 2]$ are the range and the smoothness parameters, respectively, and $\|\cdot \|$ denotes the Euclidean norm. The Smith random field with positive definite covariance matrix $\Sigma$ \citep{smith1990max} corresponds to the Brown--Resnick field associated with the semivariogram 
\Beq
\label{Eq_Variogram_Smith_Field}
\gamma_W(\bm{x})=\bm{x}' \Sigma^{-1} \bm{x}/2, \quad \bm{x} \in \mathbb{R}^d;
\Eeq
see, e.g., \cite{huser2013composite}.

If $\{ X(\bm{x}) \}_{\bm{x} \in \mathbb{R}^d}$ is max-stable, there exist functions $\eta(\cdot) \in \Mbb{R}$, $\tau(\cdot)>0$ and $\xi(\cdot) \in \Mbb{R}$ defined on $\Mbb{R}^d$, called the
location, scale and shape functions, such that
\Beq
\label{Eq_Link_Maxstb_Simple_Maxstab}
X(\bm{x}) = 
\left \{
\begin{array}{ll}
\eta(\bm{x})-\tau(\bm{x})/\xi(\bm{x}) + \tau (\bm{x})Z(\bm{x})^{\xi(\bm{x})}/\xi(\bm{x}), & \quad \xi(\bm{x}) \neq 0, \\
\eta(\bm{x}) + \tau \log Z(\bm{x}), & \quad \xi(\bm{x}) = 0,
\end{array}
\right.
\Eeq
where $\{ Z(\bm{x}) \}_{\bm{x} \in \mathbb{R}^d}$ is simple max-stable. In the following, if $\{ X(\bm{x}) \}_{\bm{x} \in \mathbb{R}^d}$ is defined by \eqref{Eq_Link_Maxstb_Simple_Maxstab} with $\{ Z(\bm{x}) \}_{\bm{x} \in \mathbb{R}^d}$ being the Brown--Resnick field associated with the semivariogram $\gamma_W$, then $X$ will be referred to as the Brown--Resnick field associated with the semivariogram $\gamma_W$ and with GEV functions $\eta(\bm{x})$, $\tau(\bm{x})$ and $\xi(\bm{x})$. If, for all $\bm{x} \in \Mbb{R}^d$, $\eta(\bm{x})=\eta$, $\tau(\bm{x})=\tau$ and $\xi(\bm{x})=\xi$, then $X$ will be termed the Brown--Resnick field associated with the semivariogram $\gamma_W$ and with GEV parameters $\eta$, $\tau$ and $\xi$.

\subsection{Theoretical contributions}
\label{Subsec_TheoreticalContribution}

Several dependence measures for max-stable vectors and fields have been introduced in the literature: the extremal coefficient \citep[e.g.,][]{schlather2003dependence}, the F-madogram \citep{cooley2006variograms} and the $\lambda$-madogram \citep{naveau2009modelling}, among others. Here we propose a new spatial dependence measure which is the correlation of powers of max-stable vectors/fields and not of max-stable vectors/fields themselves. As explained in Section \ref{Sec_Intro}, taking power transforms when using correlation is standard practice when dealing with financial time series. For $\bm{X}$ being defined by \eqref{Eq_LinkMaxstbVect_SimpleMaxstabVect} with $(Z_1, Z_2)'$ following the Hüsler--Reiss distribution \eqref{Eq_HuslerReissDist}, we study $\Mr{Corr}(X_1^{\beta_1}, X_2^{\beta_2})$, where $\beta_i \in \mathbb{N}_*$ such that $\beta_i \xi_i < 1/2$, and Corr denotes the correlation. This allows obtaining the expression of $\Mr{Corr}( X^{\beta(\bm{x}_1)}(\bm{x}_1), X^{\beta(\bm{x}_2)}(\bm{x}_2))$, $\bm{x}_1, \bm{x}_2 \in \mathbb{R}^2$, where $X$ is the Brown--Resnick field associated with any semivariogram and with GEV functions $\eta(\bm{x})$, $\tau(\bm{x})$, $\xi(\bm{x})$, and $\beta(\bm{x})$ is a function taking values in $\Mbb{N}_*$ such that $\beta(\bm{x}) \xi(\bm{x}) < 1/2$ for all $\bm{x} \in \Mbb{R}^2$. If those GEV functions and $\beta(\bm{x})$ are not spatially constant, the field $\{ X^{\beta(\bm{x})}(\bm{x}) \}_{\bm{x}\in \Mbb{R}^2}$ is not second-order stationary and its correlation function does not only depend on the lag vector. Taking constant GEV and power functions as in the case study is however reasonable when the region considered is fairly homogeneous (in terms, e.g., of altitude, weather influences and distance to a coastline) or not too large. Moreover, every non-stationary random field can be approximated by piecewise stationary fields; see \cite{koch2019SpatialRiskAxioms} and references therein. Therefore, our main focus will be on
\Beq
\label{Eq_DepMeas}
\mathcal{D}_{X, \beta}(\bm{x}_1, \bm{x}_2)=\Mr{Corr} \left( X^{\beta}(\bm{x}_1), X^{\beta}(\bm{x}_2) \right), \quad \bm{x}_1, \bm{x}_2 \in \mathbb{R}^2,
\Eeq
where $X$ is the Brown--Resnick field with GEV parameters $\eta$, $\tau$, $\xi$, and $\beta \in \mathbb{N}_*$ such that $\beta \xi < 1/2$; in this setting, $X^{\beta}$ is second-order stationary. 

Rescaled powers of max-stable random fields constitute appropriate models for the field of insured costs from high wind speeds (see Section \ref{Subsec_PowersDamageFunction} for details) and so \eqref{Eq_DepMeas} can be viewed as the correlation function of insured wind costs, thus being useful for actuarial practice. The formulas we derive in this section also make possible the estimation of the parameters of Hüsler--Reiss distributions and Brown--Resnick random fields by equalizing the theoretical correlation and the empirical one computed on the dataset (method of moments); this may be investigated in a subsequent work. Appendix \ref{Sec_Appendix_SimpleBRfield}, which deals with simple Brown--Resnick fields, can be useful in this respect.

\medskip

Before presenting the main results, we recall the importance of correlation for risk assessment in a spatial context, which justifies studying the correlation despite the existence of dependence measures specifically designed for max-stable fields. Anyway, powers of max-stable fields are not necessarily max-stable themselves, making these measures not directly usable.  

Denote by $\mathcal{C}$ the set of all real-valued and measurable\footnote{Throughout, when applied to random fields, the adjective ``measurable'' means ``jointly measurable''.} random fields on $\mathbb{R}^2$ having almost surely (a.s.) locally integrable sample paths. Furthermore, let $\mathcal{A}$ denote the set of all compact subsets of $\mathds{R}^2$ with a strictly positive Lebesgue measure and $\mathcal{A}_c$ be the set of all convex elements of $\mathcal{A}$. For any $A \in \mathcal{A}_c$, let $\bm{b}_A$ denote its barycenter and $\lambda A$ be the area obtained by applying to $A$ a homothety with center $\bm{b}_A$ and ratio $\lambda >0$.

Let $C \in \mathcal{C}$ model the insured cost per surface unit triggered by events belonging to a specific class (e.g., European windstorms) during a given period of time.
The total insured loss on a given region $A \in \mathcal{A}$ can thus be modelled by
$$ L \left(A, C \right)=\int_A C(\bm{x}) \Mr{d}\bm{x},$$
and Theorem 4 in \cite{koch2019SpatialRiskAxioms} yields
\Beq
\label{Eq_ExpressionVarianceIntegral}
\Mr{Var}\left(L\left(A, C\right)\right) = \Mr{Var}\left(C(\bm{0})\right) \int_A \int_A \Mr{Corr} \left( C(\bm{x}), C(\bm{y}) \right) \Mr{d}\bm{x} \Mr{d}\bm{y}.
\Eeq
Hence the correlation is explicitly involved in the variance of the total insured loss, which is a key quantity for an insurance company.

Moreover, assuming that $C$ belongs to $\mathcal{C}$, has a constant expectation and satisfies the CLT \citep[see][Section 3.1]{koch2017TCL} (which holds for $C = X^{\beta}$ if $X$ is the Brown--Resnick field associated with the semivariogram \eqref{Eq_Power_Variogram} and with GEV parameters $\eta$, $\tau$ and $\xi$ such that $\beta \xi < 1/2$),
$$ \sigma = \left[ \Mr{Var}\left(C(\bm{0})\right)\int_{\Mbb{R}^2} \Mr{Corr} \left( C(\bm{0}), C(\bm{x}) \right) \Mr{d}\bm{x} \right]^{1/2}$$
is the standard deviation of the normal distribution appearing in the CLT of $C$ and is thus \citep[][Theorems 2 and 5]{koch2019SpatialRiskAxioms} essential for the asymptotic distribution of $L(\lambda A, C)$ and the asymptotic properties of spatial risk measures induced by the field $C$ and associated with value-at-risk and expected shortfall. The analysis of \eqref{Eq_DepMeas} is thereby insightful for the risk assessment of wind damage; the formulas derived in this paper are used in an ongoing study.

\medskip

As \eqref{Eq_LinkMaxstbVect_SimpleMaxstabVect} specifies a transformation of simple max-stable random vectors, we first deal with such vectors. In the next theorem, we take a random vector $\bm{Z}=(Z_1, Z_2)'$ following the Hüsler--Reiss distribution \eqref{Eq_HuslerReissDist}. If $\beta \in \mathbb{R}$ and $Z$ is a standard Fr\'echet random variable, it is easily shown that $Z^{\beta}$ has a finite second moment if and only if $\beta < 1/2$, which imposes, in order for the covariance $\Mr{Cov}(Z_1^{\beta_1}, Z_2^{\beta_2})$ to exist, that $\beta_1, \beta_2 < 1/2$. This covariance and other expressions throughout this section involve, for $\beta_1, \beta_2 < 1/2$,
\Beq
\label{Eq_Def_g_beta1_beta2}
I_{\beta_1, \beta_2}(h) =
\left \{
\begin{array}{ll}
\Gamma(1-\beta_1-\beta_2), & \mbox{if} \quad  h=0, \\ 
\displaystyle \int_{0}^{\infty} \theta^{\beta_2} \Big[ C_2(\theta,h) \  C_1(\theta,h)^{\beta_1+\beta_2 -2} \ \Gamma(2-\beta_1-\beta_2) \\ \qquad + C_3(\theta,h) \ C_1(\theta,h)^{\beta_1+\beta_2-1} \ \Gamma(1-\beta_1-\beta_2) \Big] \mathrm{d}\theta, & \mbox{if} \quad h>0,
\end{array}
\right.
\Eeq
where $\Gamma$ denotes the gamma function, and, for $\theta, h > 0$,
\begin{align*}
C_1(\theta,h) &=   \Phi \left( \frac{h}{2}+ \frac{\log \theta}{h} \right)+\frac{1}{\theta} \Phi \left( \frac{h}{2}- \frac{\log \theta}{h} \right), \\
C_2(\theta,h) &= \left[   \Phi \left( \frac{h}{2}+ \frac{\log \theta}{h} \right) +\frac{1}{h} \phi \left( \frac{h}{2}+ \frac{\log \theta}{h} \right)-\frac{1}{h \theta} \phi \left(  \frac{h}{2}-\frac{\log \theta}{h} \right) \right]
\\& \quad \ \times \left[ \frac{1}{\theta^2} \Phi \left(  \frac{h}{2}- \frac{\log \theta}{h} \right)+\frac{1}{h \theta^2} \phi \left( \frac{h}{2}- \frac{\log \theta}{h} \right)-\frac{1}{h \theta} \phi \left( \frac{h}{2}+ \frac{\log \theta}{h} \right) \right], \\
C_3(\theta,h) &= \frac{1}{h^2 \theta} \left( \frac{h}{2}- \frac{\log \theta}{h} \right) \ \phi \left( \frac{h}{2}+ \frac{\log \theta}{h} \right)+\frac{1}{h^2 \theta^2} \left( \frac{h}{2}+ \frac{ \log \theta}{h} \right) \phi \left( \frac{h}{2}- \frac{\log \theta}{h}  \right),
\end{align*}
with $\Phi$ and $\phi$ denoting the standard Gaussian distribution and density functions, respectively. 

In order to obtain the next result, we take advantage of the radius/angle decomposition of multivariate extreme-value distributions.
\begin{Th}
\label{Th_CovarianceSimpleHR}
Let $\bm{Z}=(Z_1, Z_2)'$ follow the Hüsler--Reiss distribution \eqref{Eq_HuslerReissDist} with parameter $h$. Then, for all $\beta_1, \beta_2 < 1/2$,
\Beq
\label{Eq_CovarianceSimpleHR}
\mathrm{Cov} \left( Z_1^{\beta_1}, Z_2^{\beta_2} \right) = I_{\beta_1, \beta_2} \left( h \right)- \Gamma(1-\beta_1) \Gamma(1-\beta_2).
\Eeq
\end{Th} 
\begin{Rq}
Theorem \ref{Th_CovarianceSimpleHR}, which is a cornerstone of this section, stems from unpublished work in Section 4.5.1 of the PhD thesis by \cite{kochphd2014}.
\end{Rq}
We adapt Theorem \ref{Th_CovarianceSimpleHR} to the more realistic setting where the margins are general GEV distributions with non-zero shape parameters. The support of such margins possibly includes strictly negative values, and we thus consider powers which are strictly positive integers.
\begin{Th}
\label{Th_Cov_Maxstab_Real_Marg}
Let $\bm{Z}$ having \eqref{Eq_HuslerReissDist} as distribution function with parameter $h$, and let $\bm{X}=(X_1, X_2)'$ be the transformed version of $\bm{Z}$ by \eqref{Eq_LinkMaxstbVect_SimpleMaxstabVect} with $\eta_i \in \mathbb{R}$, $\tau_i >0$ and $\xi_i \neq 0$, $i=1, 2$. Moreover, let $\beta_i \in \mathbb{N}_*$ such that $\beta_i \xi_i < 1/2$, $i=1,2$.
Then,
\begin{align}
\Mr{Cov} \left( X_1^{\beta_1}, X_2^{\beta_2} \right)
&= \sum_{k_1=0}^{\beta_1} \sum_{k_2=0}^{\beta_2}  B_{k_1, \beta_1, \eta_1, \tau_1, \xi_1, k_2, \beta_2, \eta_2, \tau_2, \xi_2} \ I_{(\beta_1-k_1)\xi_1, (\beta_2-k_2) \xi_2} \left( h\right) \nonumber
\\& \quad - \sum_{k_1=0}^{\beta_1} \sum_{k_2=0}^{\beta_2} B_{k_1, \beta_1, \eta_1, \tau_1, \xi_1, k_2, \beta_2, \eta_2, \tau_2, \xi_2} \ \Gamma(1-[\beta_1-k_1]\xi_1) \Gamma(1-[\beta_2-k_2]\xi_2),
\label{EqCovHRDiffBetaMar}
\end{align}
where 
$$ B_{k_1, \beta_1, \eta_1, \tau_1, \xi_1, k_2, \beta_2, \eta_2, \tau_2, \xi_2} = {\beta_1 \choose k_1} \left( \eta_1-\frac{\tau_1}{\xi_1} \right)^{k_1} \left( \frac{\tau_1}{\xi_1} \right)^{\beta_1-k_1} {\beta_2 \choose k_2} \left( \eta_2-\frac{\tau_2}{\xi_2} \right)^{k_2} \left( \frac{\tau_2}{\xi_2} \right)^{\beta_2-k_2},$$
and, for $i=1, 2$,
\Beq
\Mr{Var} \left( X_i^{\beta_i} \right) = \sum_{k_1=0}^{\beta_i} \sum_{k_2=0}^{\beta_i} B_{k_1, k_2, \beta_i, \eta_i, \tau_i, \xi_i} \left \{ \Gamma(1-\xi_i[2 \beta_i - k_1 -k_2])- \Gamma(1-[\beta_i-k_1]\xi_i) \Gamma(1-[\beta_i-k_2]\xi_i) \right \},
\label{Eq_Var_GEV_beta_i}
\Eeq
where, for $\eta \in \mathbb{R}$, $\tau >0$, $\xi \neq 0$, and $\beta \in \mathbb{N}_*$ such that $\beta \xi < 1/2$,
$$ B_{k_1, k_2, \beta, \eta, \tau, \xi}= {\beta \choose k_1} {\beta \choose k_2} \left( \eta-\frac{\tau}{\xi} \right)^{k_1+k_2} \left( \frac{\tau}{\xi} \right)^{2\beta-(k_1+k_2)}.$$
\end{Th}
The combination of \eqref{EqCovHRDiffBetaMar} and \eqref{Eq_Var_GEV_beta_i} immediately yields the expression of $\Mr{Corr}( X_1^{\beta_1}, X_2^{\beta_2})$. We have assumed in Theorem \ref{Th_Cov_Maxstab_Real_Marg} that $\xi_i \neq 0$ but, as shown now, the case $\xi_1=\xi_2=0$ is easily recovered by taking $\xi_1=\xi_2=\xi$ and letting $\xi$ tend to $0$ in \eqref{EqCovHRDiffBetaMar}.
\begin{Prop}
\label{Prop_Continuity_Cov_xi_0}
Let $\beta_1, \beta_2 \in \mathbb{N}_*$, $\varepsilon>0$ and $S_{\beta_1, \beta_2, \varepsilon} =  \{ \xi \neq 0:  \xi <\min \{ 1 /[2 \beta_1 (1+\varepsilon)], 1 /[2 \beta_2 (1+\varepsilon)] \} \}$. Let $\bm{Z}$ be a simple max-stable vector with continuous exponent function and let $\bm{X}_{\xi}=(X_{1, \xi}, X_{2, \xi})'$ be the transformed version of $\bm{Z}$ by \eqref{Eq_LinkMaxstbVect_SimpleMaxstabVect} with $\eta_i \in \mathbb{R}$, $\tau_i >0$ and $\xi_i = \xi \in S_{\beta_1, \beta_2, \epsilon}$, $i=1, 2$. Let $\bm{X}_{0}=(X_{1, 0}, X_{2, 0})'$ be built as $\bm{X}_{\xi}$ but with $\xi=0$. 
Then,
$$ \lim_{\xi \to 0} \mathrm{Cov} \left( X_{1, \xi}^{\beta_1},  X_{2, \xi}^{\beta_2} \right) =  \mathrm{Cov} \left( X_{1, 0}^{\beta_1}, X_{2, 0}^{\beta_2} \right).$$
\end{Prop}
Using similar arguments, we get $\lim_{\xi \to 0} \mathrm{Var}( X_{i, \xi}^{\beta_i})=\mathrm{Var}( X_{i, 0}^{\beta_i})$, which yields
$$ \lim_{\xi \to 0} \mathrm{Corr} \left( X_{1, \xi}^{\beta_1},  X_{2, \xi}^{\beta_2} \right) =  \mathrm{Corr} \left( X_{1, 0}^{\beta_1},  X_{2, 0}^{\beta_2} \right).$$
This result obviously applies if $\bm{Z}$ follows the Hüsler--Reiss distribution  \eqref{Eq_HuslerReissDist}. 

Next proposition, which is an immediate corollary of Theorem \ref{Th_Cov_Maxstab_Real_Marg}, provides all the necessary ingredients for the computation of our dependence measure $\mathcal{D}_{X, \beta}$ in \eqref{Eq_DepMeas}.
\begin{Prop}
\label{EqCovPowersHRSameMarginsSameBeta}
Under the same assumptions as in Theorem \ref{Th_Cov_Maxstab_Real_Marg} but with $\eta_1=\eta_2=\eta$, $\tau_1=\tau_2=\tau$, $\xi_1=\xi_2=\xi$ and $\beta_1=\beta_2=\beta$, we have
\Beq
\Mr{Cov} \left( X_1^{\beta}, X_2^{\beta} \right)
= g_{\beta, \eta, \tau, \xi} \left( h \right) - \sum_{k_1=0}^{\beta} \sum_{k_2=0}^{\beta} B_{k_1, k_2, \beta, \eta, \tau, \xi} \ \Gamma(1-[\beta-k_1]\xi) \Gamma(1-[\beta-k_2]\xi),
\label{Eq_Cov_Maxstab_Real_Marg_Eq_Coeff}
\Eeq
with
\Beq
g_{\beta, \eta, \tau, \xi}(h) = \sum_{k_1=0}^{\beta} \sum_{k_2=0}^{\beta} B_{k_1, k_2, \beta, \eta, \tau, \xi} \ I_{(\beta-k_1)\xi, (\beta-k_2) \xi} \left( h \right),
\label{Eq_Function_gtilde}
\Eeq
and, for $i=1, 2$,
\Beq
\Mr{Var} \left( X_i^{\beta} \right) = \sum_{k_1=0}^{\beta} \sum_{k_2=0}^{\beta} B_{k_1, k_2, \beta, \eta, \tau, \xi} \left \{ \Gamma(1-\xi[2 \beta - k_1 -k_2])- \Gamma(1-[\beta-k_1]\xi) \Gamma(1-[\beta-k_2]\xi) \right \}.
\label{Eq_Var_GEV_beta}
\Eeq
\end{Prop}
The following theorem, which is a direct consequence of \eqref{EqBivDistFuncBRField} and Proposition \ref{EqCovPowersHRSameMarginsSameBeta}, gives the expression of $\mathcal{D}_{X, \beta}$.
\begin{Th}
\label{ThAppBR}
Let $X$ be the Brown--Resnick field associated with the semivariogram $\gamma_W$ and with GEV parameters $\eta \in \Mbb{R}$, $\tau>0$, $\xi \neq 0$, and let $\beta \in \mathbb{N}_*$ such that $\beta \xi < 1/2$. Then 
\Beq
\label{Eq_ExpressionMainQuantityPaper}
\mathcal{D}_{X, \beta}(\bm{x}_1, \bm{x}_2)=\Mr{Cov} \left( X^{\beta}(\bm{x}_1), X^{\beta}(\bm{x}_2) \right)/\Mr{Var}(X^{\beta}(\bm{0})), \quad \bm{x}_1, \bm{x}_2 \in \mathbb{R}^2,
\Eeq
where $\Mr{Cov} \left( X^{\beta}(\bm{x}_1), X^{\beta}(\bm{x}_2) \right)$ is given by \eqref{Eq_Cov_Maxstab_Real_Marg_Eq_Coeff} with $h=\sqrt{2 \gamma_W(\bm{x}_2-\bm{x}_1)}$ and $\Mr{Var}(X^{\beta}(\bm{0}))$ is given by \eqref{Eq_Var_GEV_beta}. 
\end{Th}
Note that the case $\xi=0$ is easily recovered as explained above.
\begin{Rq}
\label{RqCorrBRDiffMarPow}
The combination of \eqref{EqBivDistFuncBRField} and Theorem \ref{Th_Cov_Maxstab_Real_Marg} yields the following more general result than Theorem \ref{ThAppBR}. Let $\{ X(\bm{x}) \}_{\bm{x} \in \mathbb{R}^2}$ be the Brown--Resnick field associated with the semivariogram $\gamma_W$ and with GEV functions $\eta(\bm{x}) \in \Mbb{R}$, $\tau(\bm{x})>0$, $\xi(\bm{x}) \neq 0$, and let $\beta(\bm{x})$ be a function taking values in $\Mbb{N}_*$ such that $\beta(\bm{x}) \xi(\bm{x}) < 1/2$ for any $\bm{x} \in \Mbb{R}^2$. Then,
\Beq
\label{Eq_CorrVaryingBetaEtaTauXi}
\Mr{Corr} \left( X^{\beta(\bm{x}_1)}(\bm{x}_1), X^{\beta(\bm{x}_2)}(\bm{x}_2) \right)= \frac{\Mr{Cov} \left( X^{\beta(\bm{x}_1)}(\bm{x}_1), X^{\beta(\bm{x}_2)}(\bm{x}_2) \right)}{\sqrt{\Mr{Var}(X^{\beta(\bm{x}_1)}(\bm{x}_1)) \Mr{Var}(X^{\beta(\bm{x}_2)}(\bm{x}_2))}}, \quad \bm{x}_1, \bm{x}_2 \in \Mbb{R}^2,
\Eeq
where $\Mr{Cov} \left( X^{\beta(\bm{x}_1)}(\bm{x}_1), X^{\beta(\bm{x}_2)}(\bm{x}_2) \right)$ is given by \eqref{EqCovHRDiffBetaMar} with $h=\sqrt{2 \gamma_W(\bm{x}_2-\bm{x}_1)}$, $\eta_i = \eta(\bm{x}_i), \tau_i = \tau(\bm{x}_i), \xi_i = \xi(\bm{x}_i)$, $\beta_i = \beta(\bm{x}_i)$, $i=1,2$, and $\Mr{Var}(X^{\beta(\bm{x}_i)}(\bm{x}_i))$ is given by \eqref{Eq_Var_GEV_beta} with $\eta_i = \eta(\bm{x}_i), \tau_i = \tau(\bm{x}_i), \xi_i = \xi(\bm{x}_i)$, $\beta_i = \beta(\bm{x}_i)$. 
\end{Rq}
The analytical formulas in Theorems \ref{Th_CovarianceSimpleHR},  \ref{Th_Cov_Maxstab_Real_Marg}, \ref{ThAppBR}, Proposition \ref{EqCovPowersHRSameMarginsSameBeta}, and Remark \ref{RqCorrBRDiffMarPow} allow a more accurate and much faster computation of the respective quantities than using Monte Carlo methods as the involved integrals can be computed fast and with high precision using, e.g., adaptive quadrature. The Smith field being a member of the class of Brown--Resnick fields, Theorem \ref{ThAppBR} and Remark \ref{RqCorrBRDiffMarPow} also apply for $X$ being the Smith field with any covariance matrix.

The influence of the marginal parameters and of the power $\beta$ merits some theoretical comments.
Let $\bm{Z}=(Z_1, Z_2)'$ and $\bm{X}=(X_1, X_2)'$ be as in Theorem \ref{Th_Cov_Maxstab_Real_Marg} and suppose that $X_1$ and $X_2$ are a.s. strictly positive (i.e., $\xi_1, \xi_2>0$ and $\eta_1 - \tau_1/\xi_1, \eta_2 - \tau_2/\xi_2>0$). For $\eta \in \Mbb{R}$, $\tau >0$ and $\xi \neq 0$, the transformation $z \mapsto \eta -\tau /\xi + \tau z^{\xi}/\xi$, $z>0$, is strictly increasing and the same applies for $x \mapsto x^{\beta}$, $x>0$, with $\beta \in \Mbb{N}$, and $z \mapsto z^{\beta^*}$, $z>0$, with $0 < \beta^* < 1/2$. Thus, owing to the invariance of the copula of a distribution under
strictly increasing transformations of the margins, the copula of $(X_1^{\beta_1}, X_2^{\beta_2})'$ is the same whatever the values of $\beta_i \in \Mbb{N}_*$, and is the same as the copula of $(Z_1^{\beta_1^*}, Z_2^{\beta_2^*})'$ whatever the values of $\beta_i^*$ such that $0 < \beta_i^* < 1/2$. However, the correlation between two random variables does not only depend on their copula but also on their margins, and is typically not invariant under non-linear transformations. We do not have equality between $\mathrm{Corr} ( X_1^{\beta_1}, X_2^{\beta_2})$ and $\mathrm{Corr} (Z_1^{\beta_1^*}, Z_2^{\beta_2^*})$ in general, as can also be seen directly from the formulas, and this also holds in the particular case where $\bm{Z}$, $\bm{X}$ and $\beta_1, \beta_2$ are as in Proposition \ref{EqCovPowersHRSameMarginsSameBeta} and $\beta_1^* = \beta_2^* = \beta^*$ such that $0 <\beta^*< 1/2$. We have 
$\mathrm{Corr} ( X_1^{\beta}, X_2^{\beta})\neq \mathrm{Corr} (Z_1^{\beta^*}, Z_2^{\beta^*})$ and, for $\beta \neq 1$, $\mathrm{Corr} ( X_1^{\beta}, X_2^{\beta})\neq \mathrm{Corr} ( X_1, X_2)$. Thus, $\mathcal{D}_{X, \beta}$ in \eqref{Eq_DepMeas} is not invariant with respect to the marginal parameters $\eta$, $\tau$, $\xi$ and the power $\beta$. Taking the appropriate values of those quantities is necessary when using $\mathcal{D}_{X, \beta}$ for concrete risk assessment problems, and studying its sensitivity with respect to $\beta$ is also of interest. The conclusions of this paragraph regarding the correlations are a fortiori true if $X_1, X_2$ are not a.s. strictly positive; in that case, even the mentioned equalities of copulas do not hold in general.

We now investigate the behaviour of the function $g_{\beta, \eta, \tau, \xi}$ defined in \eqref{Eq_Function_gtilde} in order to derive useful conclusions about $\mathcal{D}_{X, \beta}$ and because we need it in an ongoing work about spatial risk measures. The proof of next proposition is appealing as it first involves showing a result (Proposition \ref{Prop_Generalization_Dhaene} in Appendix \ref{subsubsec_Generalization_Dhaene}) about the correlation order, which is a classical concept of dependence comparison in actuarial risk theory \citep[e.g.,][Section 6.2]{denuit2005actuarial}.
\begin{Prop}
\label{Prop_Decrease_gtilde}
For all $\eta \in \mathbb{R}$, $\tau>0$, $\xi \neq 0$ and $\beta \in \mathbb{N}_*$ such that $\beta \xi < 1/2$, the function $g_{\beta, \eta, \tau, \xi}$ defined in \eqref{Eq_Function_gtilde} is strictly decreasing.
\end{Prop}
The two following propositions state the continuity of $g_{\beta, \eta, \tau, \xi}$ and characterize its behaviour around $0$ and at $\infty$.
\begin{Prop}
\label{Prop_Lim_gtildebeta_hto0}
For all $\eta \in \mathbb{R}$, $\tau>0$, $\xi \neq 0$ and $\beta \in \mathbb{N}_*$ such that $\beta \xi < 1/2$, the function $g_{\beta, \eta, \tau, \xi}$ defined in \eqref{Eq_Function_gtilde} satisfies
\Beq
\label{Eq_Lim_gtildebeta_hto0}
\lim_{h \to 0} g_{\beta, \eta, \tau, \xi}(h) = \sum_{k_1=0}^{\beta} \sum_{k_2=0}^{\beta} B_{k_1, k_2, \beta, \eta, \tau, \xi} \Gamma(1-\xi[2 \beta - k_1 -k_2])
\Eeq
and is continuous everywhere on $[0, \infty)$.
\end{Prop}
\begin{Prop}
\label{Prop_Lim_gtilde_infty}
For all $\eta \in \mathbb{R}$, $\tau>0$, $\xi \neq 0$ and $\beta \in \mathbb{N}_*$ such that $\beta \xi < 1/2$, the function $g_{\beta, \eta, \tau, \xi}$ defined in \eqref{Eq_Function_gtilde} satisfies
\Beq
\label{Eq_Lim_gtilde_infty}
\lim_{h \to \infty} g_{\beta, \eta, \tau, \xi}(h)=\sum_{k_1=0}^{\beta} \sum_{k_2=0}^{\beta} B_{k_1, k_2, \beta, \eta, \tau, \xi} \ \Gamma(1-[\beta-k_1]\xi) \Gamma(1-[\beta-k_2]\xi).
\Eeq
\end{Prop}
By Theorem \ref{ThAppBR}, $\mathcal{D}_{X, \beta}(\bm{x}_1, \bm{x}_2)$ depends on $\bm{x}_1$ and $\bm{x}_2$ through $\gamma_W(\bm{x}_2 - \bm{x}_1)$ only. As a variogram is a non-negative conditionally negative definite function, it follows from \citet[][Chapter 4, Section 3, Proposition 3.3]{berg1984harmonic}\footnote{In that book, the term ``non-negative'' is used for ``conditionally non-negative''.} that $d(\bm{x}_1, \bm{x}_2)= \sqrt{2 \gamma_W(\bm{x}_2-\bm{x}_1)}$, $\bm{x}_1, \bm{x}_2 \in \mathbb{R}^2$, defines a metric. For many common models of isotropic semivariogram $\gamma_W$, $\gamma_W(\bm{x}_2-\bm{x}_1)$ is a strictly increasing function of $\| \bm{x}_2-\bm{x}_1 \|$, which implies by \eqref{Eq_ExpressionMainQuantityPaper} and Proposition \ref{Prop_Decrease_gtilde} that $\mathcal{D}_{X, \beta}(\bm{x}_1, \bm{x}_2)$ is a strictly decreasing function of $\| \bm{x}_2-\bm{x}_1 \|$; such a decrease of the correlation with the distance seems natural. Moreover \eqref{Eq_ExpressionMainQuantityPaper} and \eqref{Eq_Lim_gtildebeta_hto0} give that $\lim_{ \bm{x}_2-\bm{x}_1 \to \bm{0}} \mathcal{D}_{X, \beta}(\bm{x}_1, \bm{x}_2)=1$, and \eqref{Eq_ExpressionMainQuantityPaper} and \eqref{Eq_Lim_gtilde_infty} imply, provided $\lim_{\| \bm{x}_2-\bm{x}_1 \| \to \infty} \gamma_W(\bm{x}_2-\bm{x}_1)=\infty$, that $\lim_{\| \bm{x}_2-\bm{x}_1 \| \to \infty}  \mathcal{D}_{X, \beta}(\bm{x}_1, \bm{x}_2)=0$. The faster the increase of $\gamma_W$ to infinity, the faster the convergence of $\mathcal{D}_{X, \beta}(\bm{x}_1, \bm{x}_2)$ to $0$. These results are consistent with our expectations. For a function $f$ from $\mathbb{R}^2$ to $\mathbb{R}$, by $\lim_{\| \bm{h} \| \to \infty} f(\bm{h})=\infty$, we mean $\lim_{h \to \infty} \inf_{\bm{u} \in \mathcal{B}_1} \{ f(h \bm{u}) \}=\infty$, where $\mathcal{B}_1= \{ \bm{x} \in \mathbb{R}^2: \| \bm{x} \|=1 \}$.

\section{Case study}
\label{Sec_Application}

We focus on insured losses from wind extremes for residential buildings over a large part of Germany, more precisely over the rectangle from $5.75 ^\circ$ to $12 ^\circ$ longitude and $49 ^\circ$ to $52 ^\circ$ latitude (see Figure \ref{Fig_MapGridPoints}). We apply the results developed in Section \ref{Subsec_TheoreticalContribution} for assessing the spatial dependence of those losses. For the insured cost field, we use the model introduced in \citet[][Section 2.3]{koch2017spatial}, that is
\Beq
\label{Eq_CostFieldModel}
C(\bm{x}) = E(\bm{x}) D(X(\bm{x})),\quad \bm{x} \in \mathbb{R}^2,
 \Eeq
where $E$ is the strictly positive and deterministic field of insured value per surface unit, $D: \Mbb{R} \mapsto [0,1]$ is the damage function, and $X$ is the model for the random field of the environmental variable generating risk. Applying the damage function $D$ to $X$ allows getting at each site the insured cost ratio, which, multiplied by the insured value, gives the corresponding insured cost. We assume the risk to be generated by wind speed maxima and we model the latter with a Brown--Resnick and a Smith max-stable model. Section \ref{Subsec_PowersDamageFunction} outlines and thoroughly justifies the power damage function $D$ that we will use. In Section \ref{Sec_ModelInferenceValidation} we describe the wind speed data and perform model estimation, selection and validation. Finally, we apply in Section \ref{Subsec_Results} the results of Section \ref{Subsec_TheoreticalContribution} using the derived insured cost model.

\subsection{Power damage function}
\label{Subsec_PowersDamageFunction}

We consider the damage function 
$$
D(w)=
\left \{
\begin{array}{cc}
(w/c_1)^{\beta}, & \quad w \leq c_1, \\
1, & \quad w \geq c_1,
\end{array}
\right.
$$
where $\beta \in \mathbb{N}_*$ and $c_1 >0$. The quantity $c_1$ corresponds to the wind speed above which the insured cost ratio equals unity and can be assumed to be much larger than possible wind speed values in Germany, especially as the distribution of wind speed maxima is bounded in this application (see below). In practice it is hence equivalent to take
\Beq
\label{Eq_GeneralDamageFunction}
D(w) = (w/c_1)^{\beta}, \quad w \in \mathbb{R},
\Eeq
and this is our choice in the following.

Power functions are perfectly suited to the case of wind. The total cost for a specific structure should increase as the square or the cube of the maximum wind speed since wind loads and dissipation rate of wind kinetic energy are proportional to the second and third powers of wind speed, respectively. For arguments supporting the use of the square, see, e.g., \citet[][Equations (4.7.1), (8.1.1) and (8.1.8) and the interpretation following Equation (4.1.20)]{simiu1996wind}. Regarding the cube, see, among others, \citet[][Chapter 2, p.7]{lamb1991historic} where the cube of the wind speed appears in the severity index, and \cite{emanuel2005increasing}. In his discussion of the paper by \cite{powell2007tropical}, \cite{kantha2008tropical} states that wind damage for a given structure must be proportional to the rate of work done (and not the force exerted) by the wind and therefore strongly argues in favour of the cube. In addition to this debate about whether the square or cube is more appropriate for total costs, several studies in the last two decades have found power-laws with much higher exponents when insured costs are considered. For instance, \cite{prahl2012applying} find powers ranging from $8$ to $12$ for insured losses on residential buildings in Germany (local damage functions). \cite{prahl2015comparison} argue that, if the total cost follows a cubic law but the insurance contract is triggered only when that cost exceeds a strictly positive threshold (e.g., in the presence of a deductible), then the resulting cost for the insurance company is of power-law type but with a higher exponent. We have validated this statement using simulations and observed that the resulting exponent depends on the threshold (not shown).

Several authors \citep[e.g.,][]{klawa2003model, pinto2007changing, donat2011high} use, even in the case of insured losses, a cubic relationship that they justify with the physical arguments given above. However, they apply the third power to the difference between the wind speed value and a high percentile of the wind distribution and not to the effective wind speed; as shown by \citet[][Appendix A3]{prahl2015comparison}, this is equivalent to applying a much higher power to the effective wind speed. Note that exponential damage functions are sometimes also encountered in the literature \cite[e.g.,][]{huang2001long, prettenthaler2012risk}; we do not consider such functions here.

According to \cite{prahl2012applying} who use \eqref{Eq_GeneralDamageFunction} as well, a spatially-constant exponent of $10$ seems
appropriate in our region for insured losses on residential buildings; see their Figure 2. Finally, \eqref{Eq_GeneralDamageFunction} yields $c_1 = w/D(w)^{1/\beta}$ for any $w>0$ and one reads in \citet[][Figure 1]{prahl2012applying} $D(26) \approx 10^{-5}$, leading to $c_1 \approx 82.2$ m s$^{-1}$. Our damage function is then
\Beq
\label{Eq_SpecificDamageFunction}
D(w) = (w/82.2)^{10}, \quad w \in \Mbb{R}.
\Eeq
As will be seen, the normalization does not play any role in our application.

\subsection{Wind data and model for extreme winds}
\label{Sec_ModelInferenceValidation}

\subsubsection{Wind data}

We consider hourly maxima of the $3$ s wind gust at $10$ m height (as defined by the World Meteorological Organization) from 1 January 1979 08:00 central European time (CET) to 1 January 2020 at 00:00 CET. This is publicly available data from the European Centre for Medium-Range Weather Forecasts (ECMWF); more precisely we use the ``$10$ m wind gust since previous post-processing'' variable in the ERA5 (ECMWF Reanalysis 5th Generation) dataset. The covered region is a rectangle from $5.75 ^\circ$ to $12 ^\circ$ longitude and $49 ^\circ$ to $52 ^\circ$ latitude and the resolution is $0.25 ^\circ$ latitude and $0.25 ^\circ$ longitude, leading to $338$ grid points. We randomly choose $226$ of them to fit the models and use the remaining $112$ for model validation; see Figure \ref{Fig_MapGridPoints}. This area encompasses the Ruhr region in Germany and is associated with high residential insured values per surface unit.  

\begin{figure}[!ht]
\begin{center}
\includegraphics[scale=0.7]{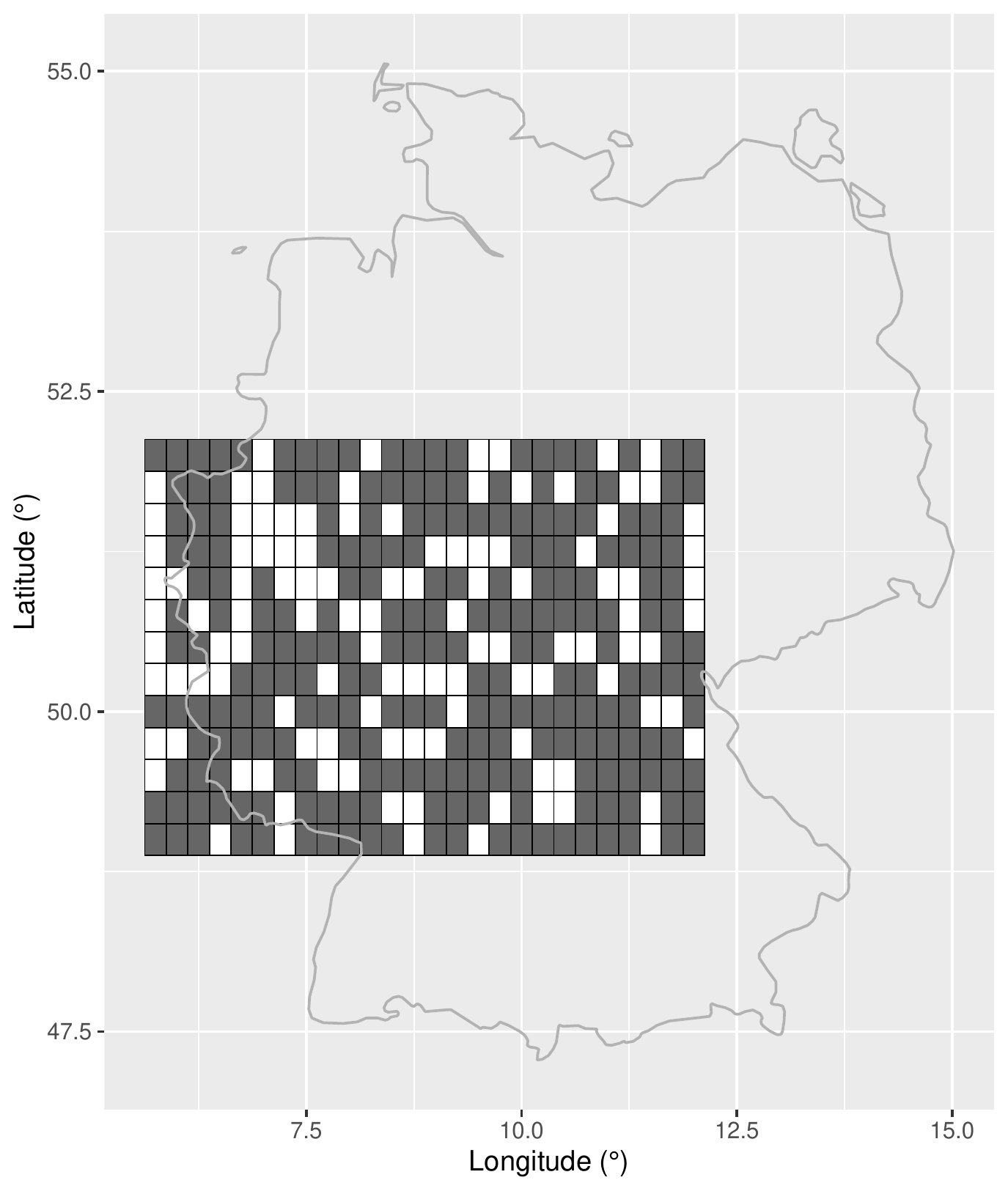}
\caption{The grey and white cells correspond to the $226$ and $118$ calibration and validation grid points, respectively.}
\label{Fig_MapGridPoints}
\end{center}
\end{figure}

We derive at each grid point the $42$ seasonal (from October to March) maxima and fit the models to the resulting pointwise maxima. For the first and last season, the maxima are computed over January--March and October--December, respectively. Focusing on October--March allows us to get rid of seasonal non-stationarity in the wind speed time series and to mainly account for winter storms rather than intense summer thunderstorms.

\subsubsection{Model}
\label{Subsubsec_Model}

We consider both the Brown--Resnick field with semivariogram \eqref{Eq_Power_Variogram} and the Smith field. As mentioned above, max-stable models are very natural ones for pointwise maxima, and the Brown--Resnick field generally shows good performance on environmental data. We model the location, scale and shape parameters as constant across the region, which is reasonable here (this can be explained by the homogeneity in terms of altitude and weather influences). Using trend surfaces for these parameters rather than fitting them separately at each grid point is standard practice as it reduces parameter uncertainty, allows a joint estimation of all marginal and dependence parameters in a reasonable amount of time and enables prediction at sites where no observations are available. Allowing anisotropy in the semivariogram of the Brown--Resnick model would be pertinent but would not modify our main conclusions. Isotropy already leads to a very satisfying model and makes our dependence measure \eqref{Eq_DepMeas} isotropic in the original space, which facilitates our discussions in Section \ref{Subsec_Results}.

Both models are fitted using maximum pairwise likelihood \citep[e.g.,][]{padoan2010likelihood} implemented in
the \texttt{fitmaxstab} function of the \texttt{SpatialExtremes} \texttt{R} package \citep{PackageSpatialExtremes}; marginal and dependence parameters were jointly estimated using the Nelder--Mead algorithm with a relative convergence tolerance of $1.49 \times 10^{-8}$. We then perform model selection by minimization of the composite likelihood information criterion (CLIC); see \cite{varin2005note}. According to that criterion, the Brown--Resnick field is the most compatible with the data; see Table \ref{Table_CLIC_Param_Application}.
\begin{table}[!h]
\center
\resizebox{\textwidth}{!}{
\begin{tabular}{c|c|c|c|c|c|c|c}
\textbf{Brown--Resnick} & CLIC &  & $\kappa $ & $\psi $ & $\eta $ & $\tau $
& $\xi $ \\ 
& 10'503'932 &  & $3.28\ (1.11)$ & $0.83\ (0.06)$ & $25.69\ (0.41)$ & $3.05\
(0.22)$ & $-0.12 \ (0.02)$ \\ \hline
\textbf{Smith} & CLIC & $\sigma _{11}$ & $\sigma _{12}$ & $\sigma _{22}$ & $%
\eta $ & $\tau $ & $\xi $ \\ 
& 10'603'208 & $4.17 \ (0.75)$ & $-0.17\ (0.05)$ & $1.03\ (0.19)$ & $25.71\
(0.37) $ & $3.07\ (0.20)$ & $-0.12\ (0.01)$
\end{tabular}
}
\newline
\caption{CLIC values and parameters' estimates (standard errors inside the brackets) of the Brown--Resnick and
Smith models.}
\label{Table_CLIC_Param_Application}
\end{table}

Figure \ref{Fig_GoodnessFitBR} shows that the theoretical pairwise extremal coefficient function of the fitted
Brown--Resnick model agrees reasonably well with the empirical pairwise extremal coefficients for the validation grid points. It is slightly above their binned estimates when those are computed using the empirical distribution functions. This small underestimation of the spatial dependence likely comes from the choice of parsimonious trend surfaces for the location, scale and shape
parameters, and disappears when we compute the empirical extremal coefficients using the marginal parameters' estimates. Overall Figure \ref{Fig_GoodnessFitBR} indicates that the proposed model fits the extremal dependence structure of the data fairly well. Figure \ref{Fig_QQplotsValidation} is complementary as it assesses both the marginal and dependence components; it shows that the distributions of several summary statistics in various dimensions are very similar for our model and the data. Finally Figure \ref{FigPathsRealSimul} suggests, for two seasons with different ranges of values, that realizations from our model have similar patterns as observed pointwise maxima, although being slightly rougher. The combination of these goodness-of-fit assessments shows that the proposed model is well-suited to our data, and so that this case study is useful in practice.
\begin{figure}[!h]
    \centering
    \begin{subfigure}[b]{0.48\textwidth}
        \centering
       \includegraphics[width=\textwidth]{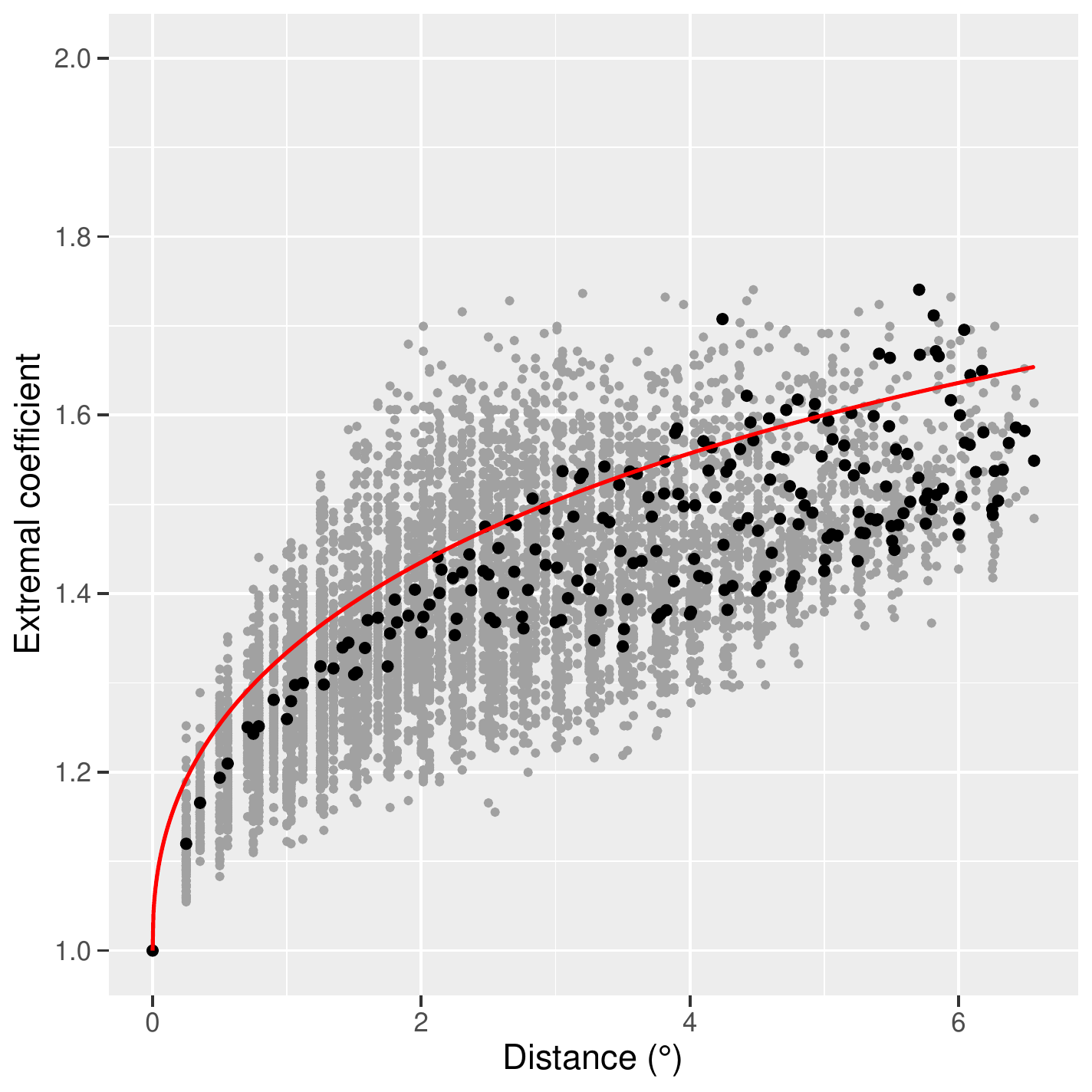}
        \end{subfigure}
        \hfill
        \begin{subfigure}[b]{0.48\textwidth}  
            \centering 
            \includegraphics[width=\textwidth]{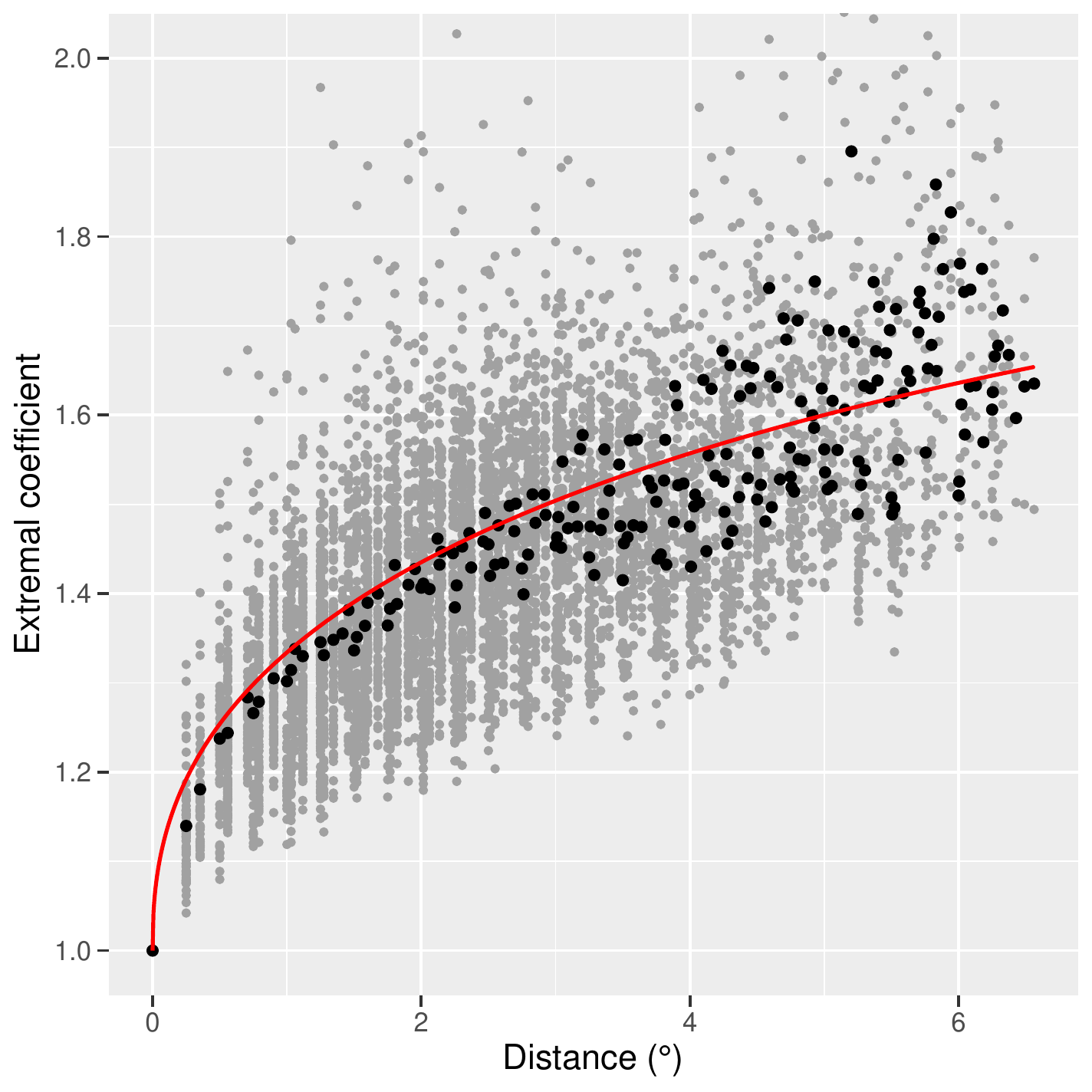}
        \end{subfigure}
        \caption{Model's performance on the validation grid points. Theoretical pairwise extremal coefficient function from the fitted Brown--Resnick model (red line) and empirical pairwise extremal coefficients (dots). The grey and black
dots are pairwise and binned estimates, respectively. The empirical extremal coefficients have been computed using the empirical distribution functions (left) and the obtained GEV parameters (right).}
\label{Fig_GoodnessFitBR}
\end{figure}
\begin{figure}[!h]
    \centering
    \begin{subfigure}[b]{0.32\textwidth}
        \centering
       \includegraphics[width=\textwidth]{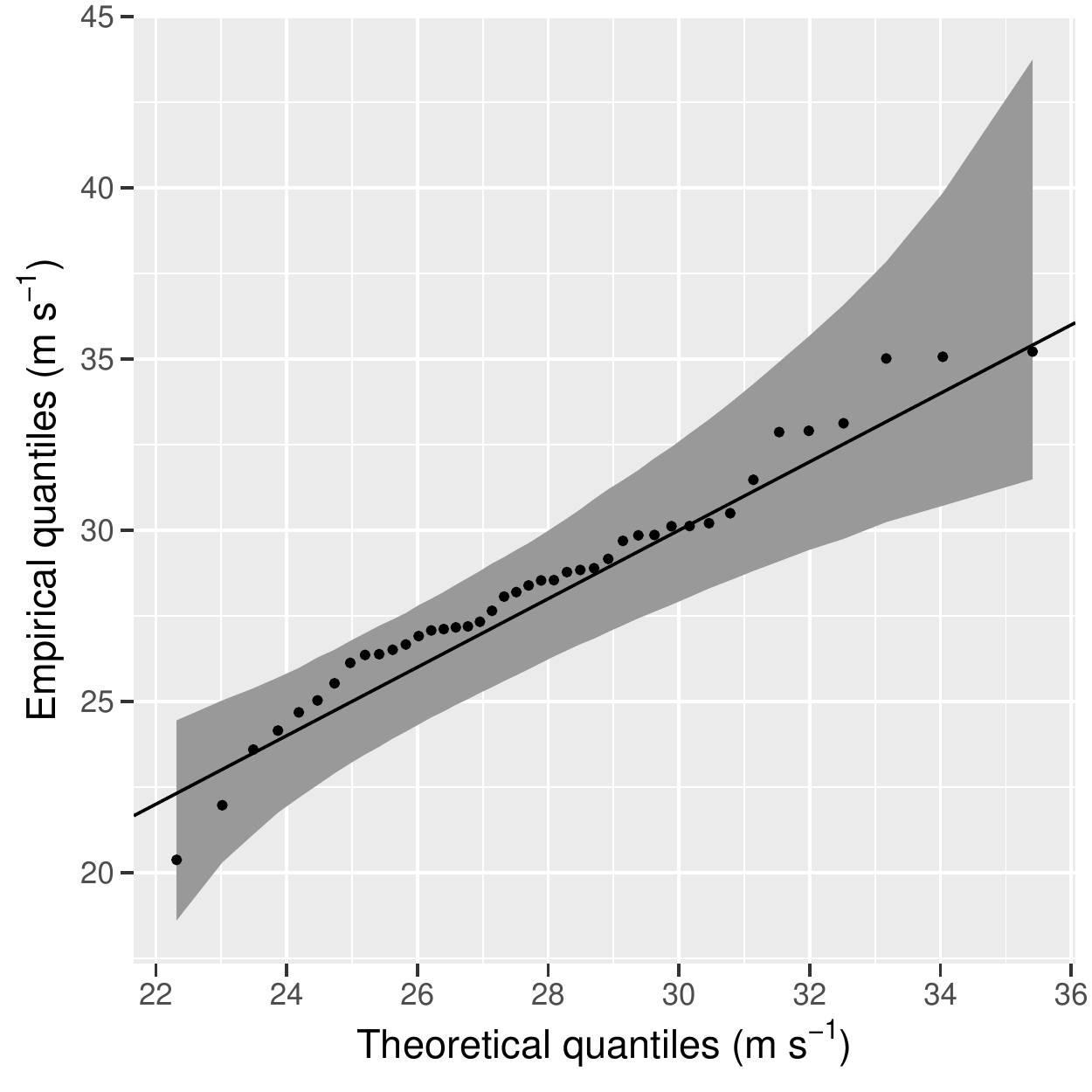}
        \end{subfigure}
        \hfill
        \begin{subfigure}[b]{0.32\textwidth}  
            \centering 
            \includegraphics[width=\textwidth]{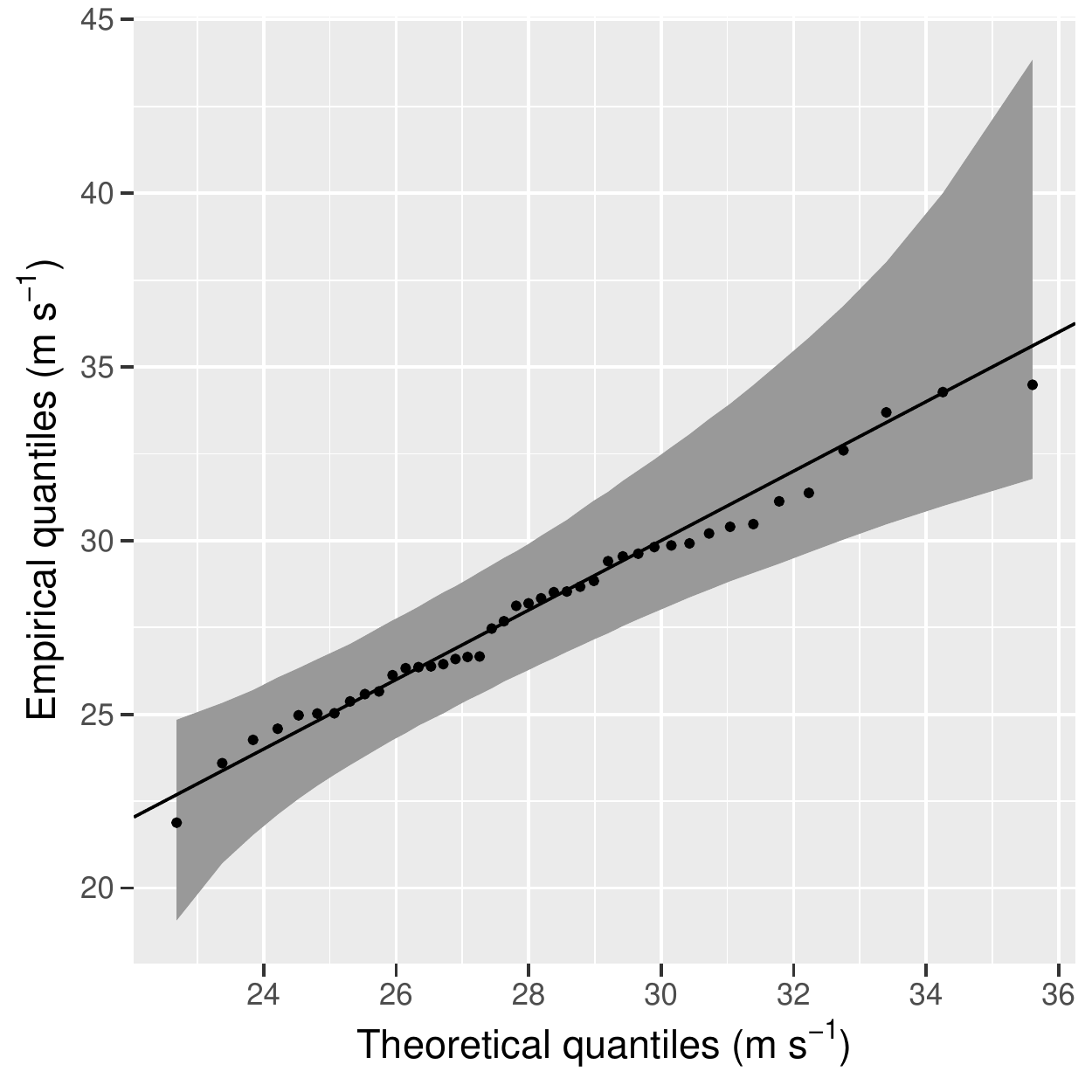}
        \end{subfigure}
        \hfill
        \begin{subfigure}[b]{0.32\textwidth}  
            \centering 
            \includegraphics[width=\textwidth]{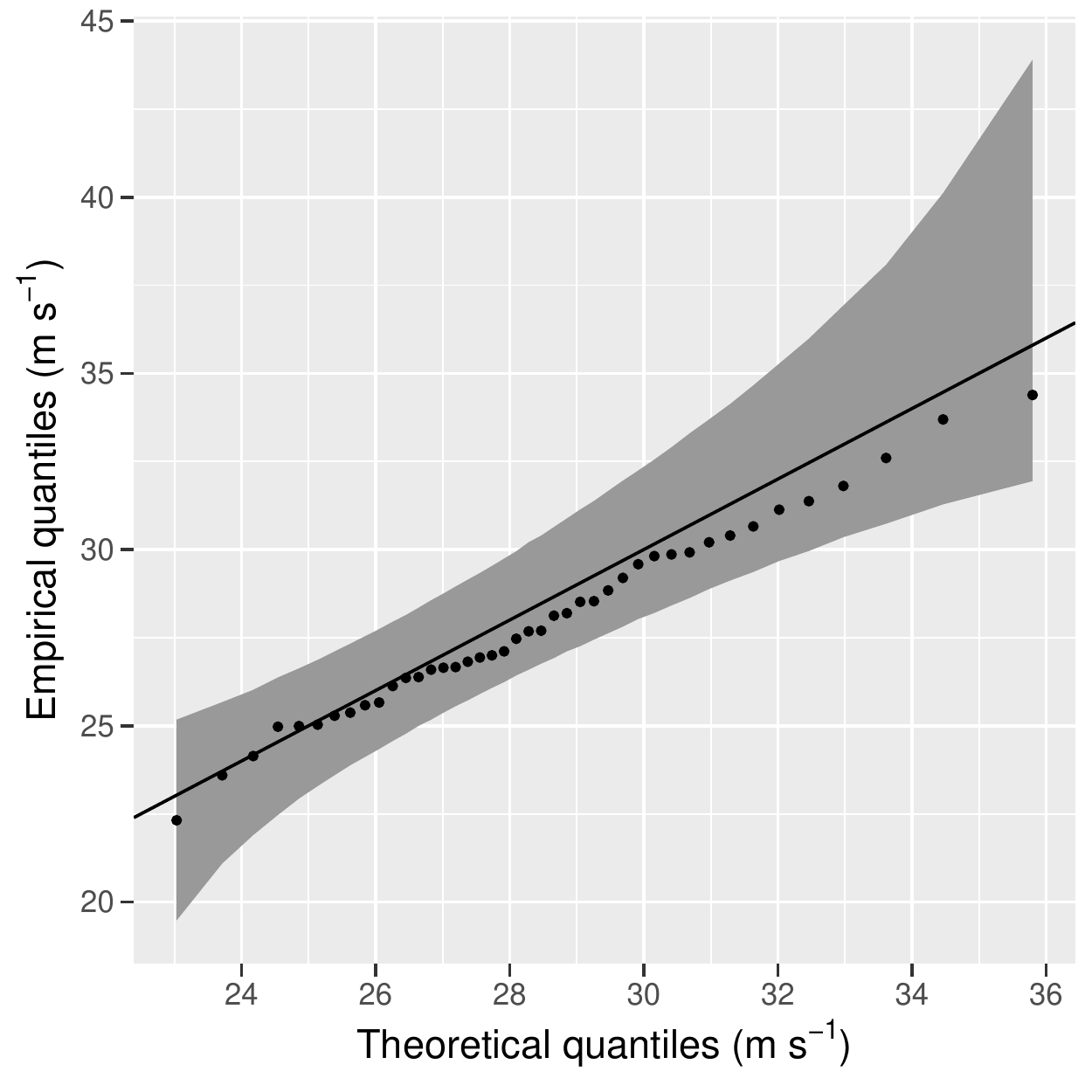}
        \end{subfigure}
        \vskip\baselineskip
        \begin{subfigure}[b]{0.32\textwidth}   
            \centering 
            \includegraphics[width=\textwidth]{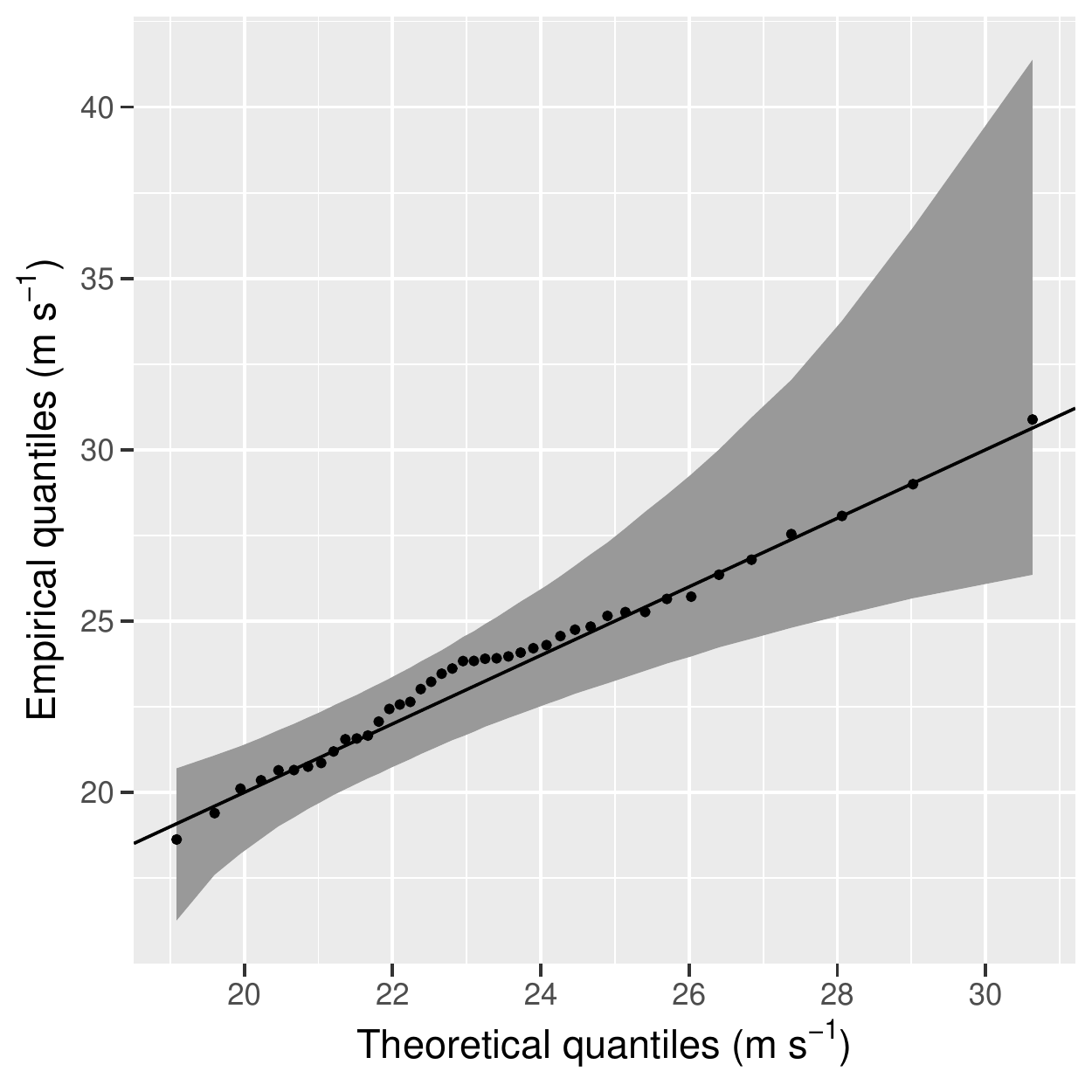}
        \end{subfigure}
        \hfill
        \begin{subfigure}[b]{0.32\textwidth}   
            \centering 
            \includegraphics[width=\textwidth]{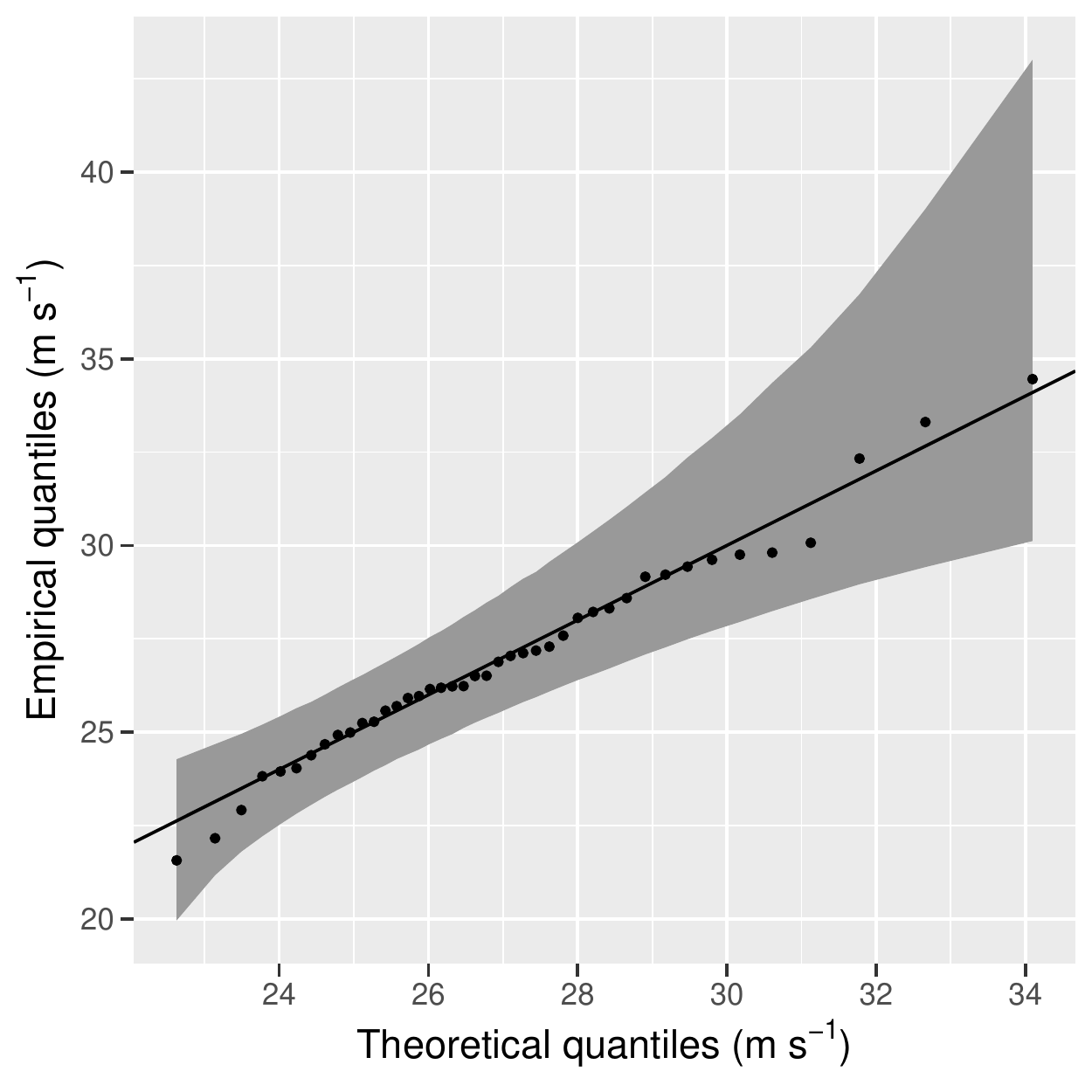}
        \end{subfigure}
        \hfill
        \begin{subfigure}[b]{0.32\textwidth}   
            \centering 
            \includegraphics[width=\textwidth]{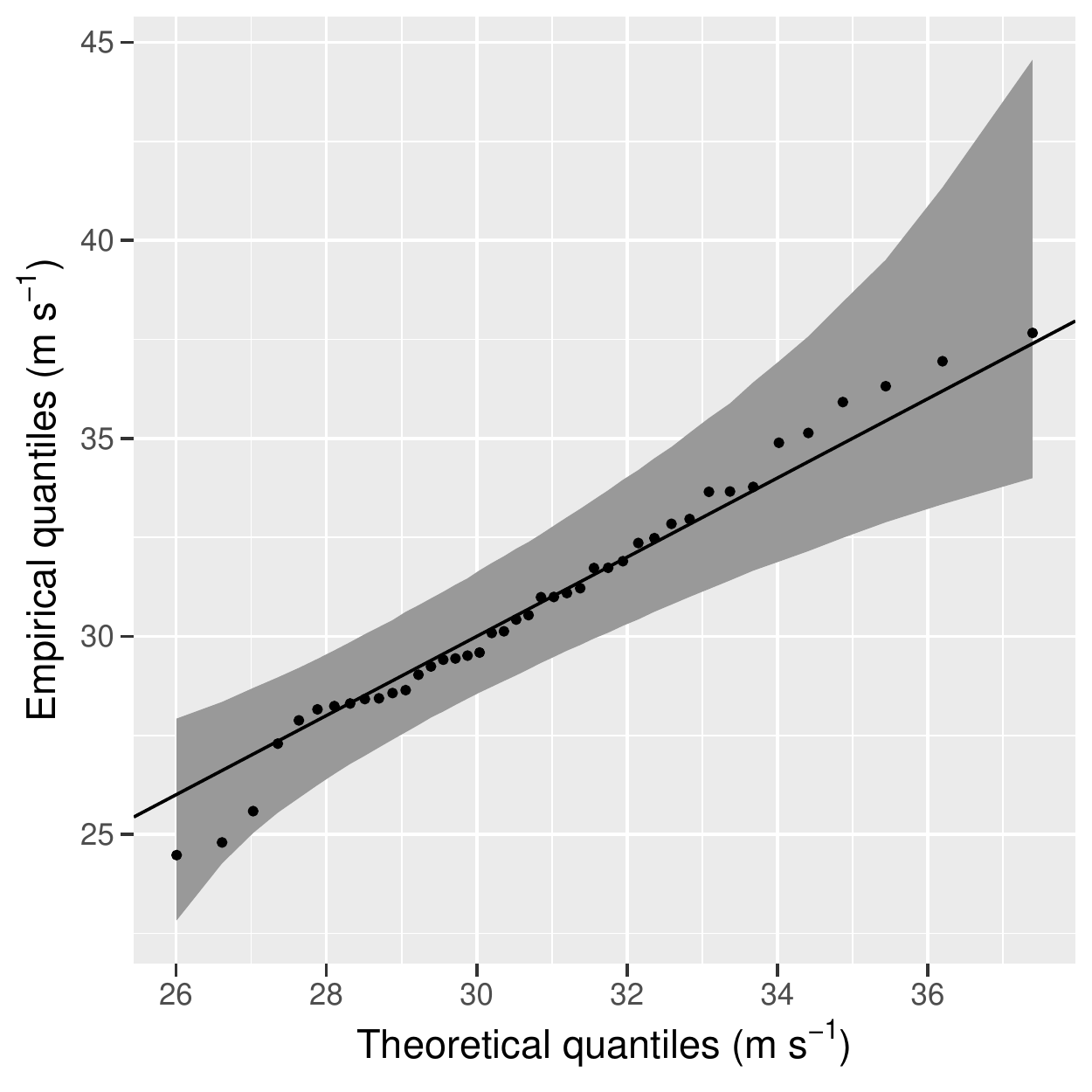}
        \end{subfigure}
        \vskip\baselineskip
        \begin{subfigure}[b]{0.32\textwidth}   
            \centering 
            \includegraphics[width=\textwidth]{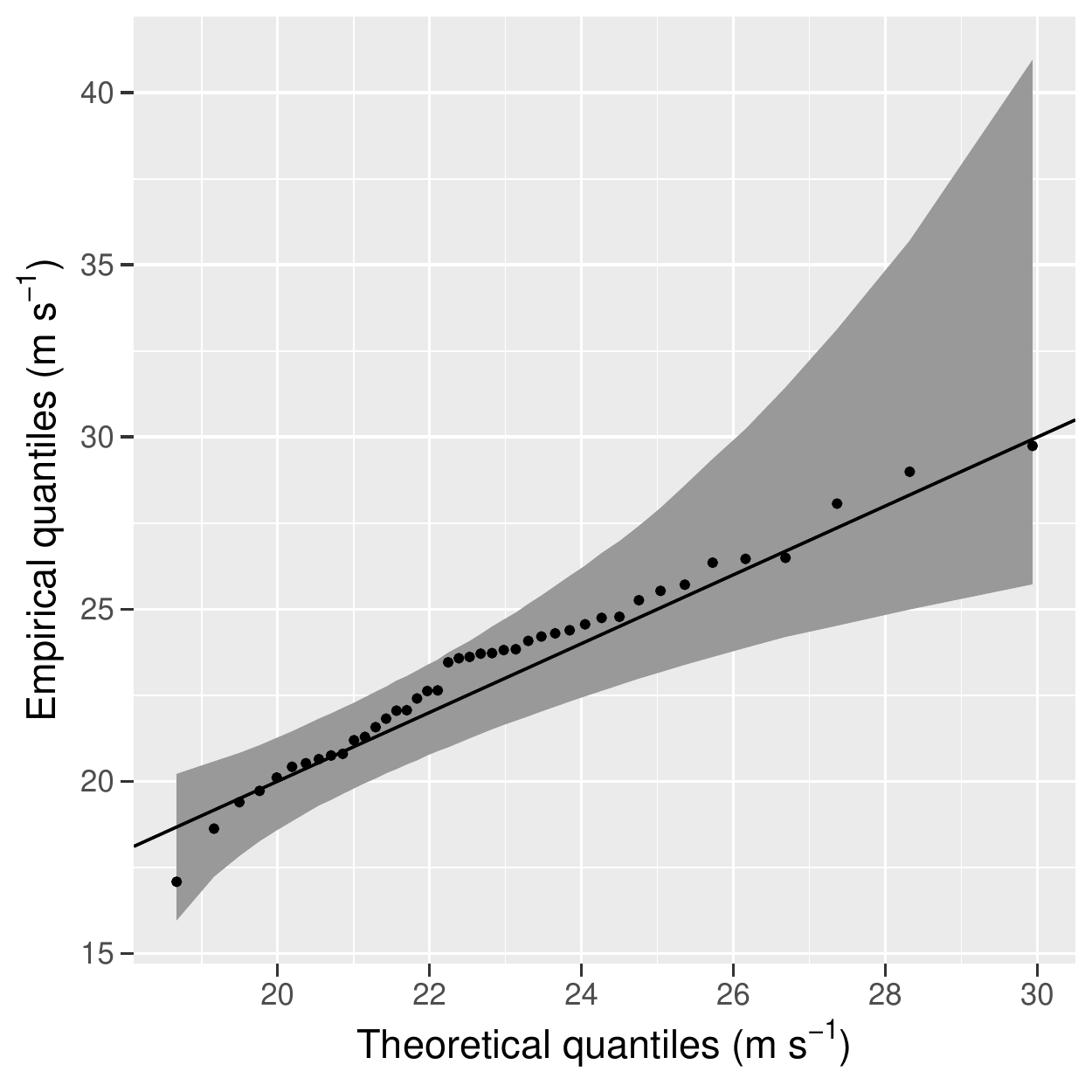}
        \end{subfigure}
        \hfill
        \begin{subfigure}[b]{0.32\textwidth}   
            \centering 
            \includegraphics[width=\textwidth]{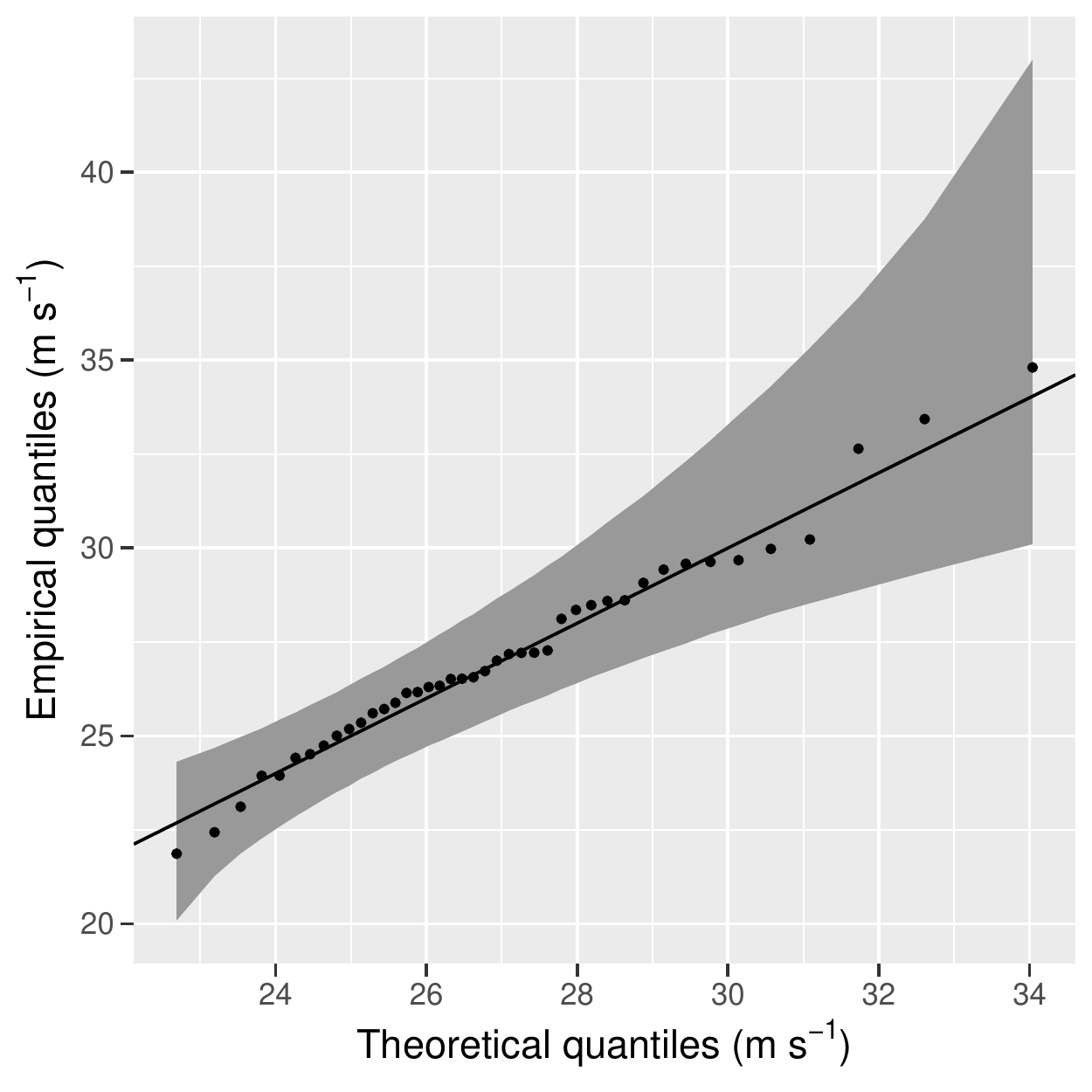}
        \end{subfigure}
        \hfill
        \begin{subfigure}[b]{0.32\textwidth}   
            \centering 
            \includegraphics[width=\textwidth]{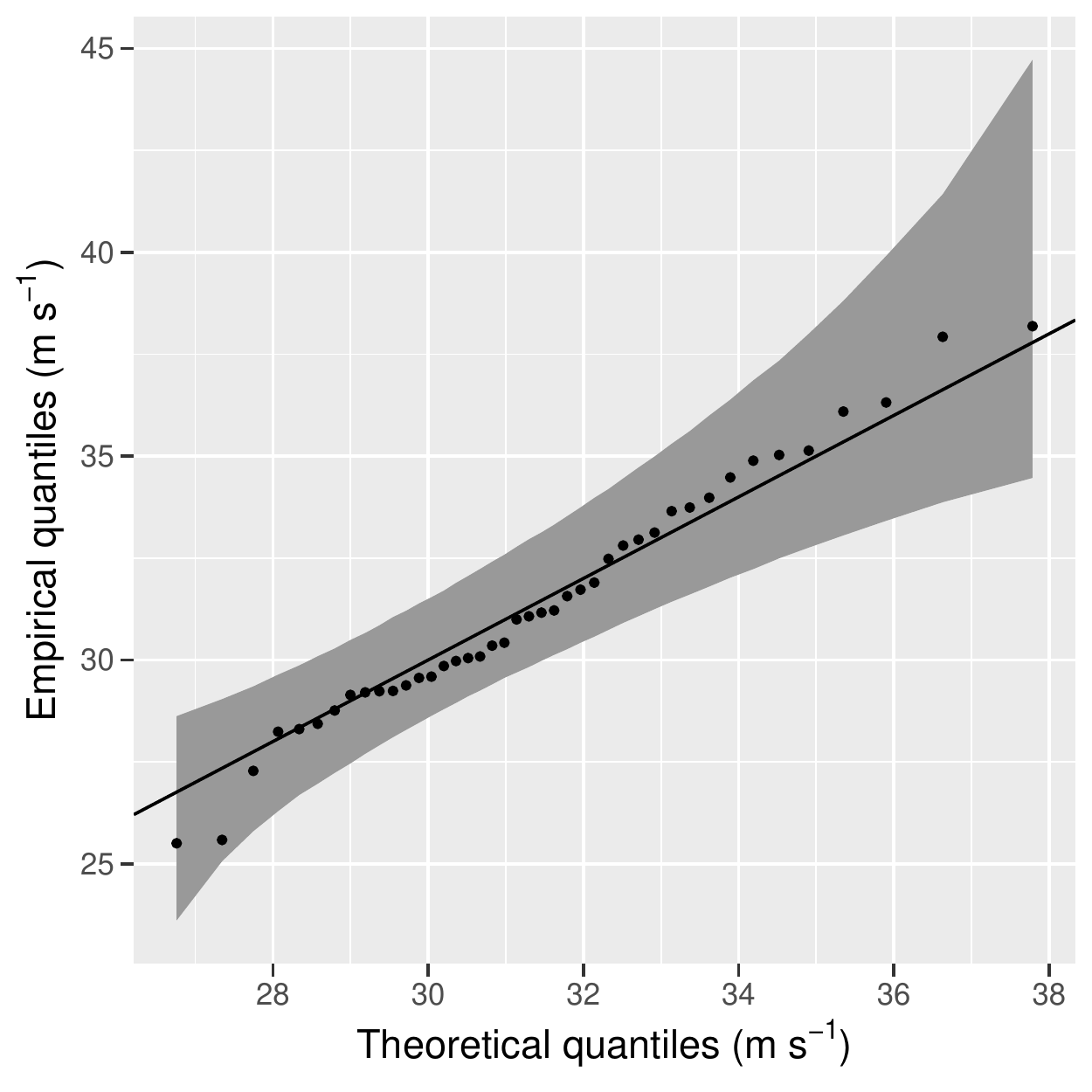}
        \end{subfigure}
        \caption{Performance of the fitted Brown--Resnick model on the validation grid points. The top row concerns maxima for pairs of validation grid points separated by a low (left), moderate (middle) and long (right) distance. The middle row focuses on minima (left), mean (middle) and maxima (right) for a group of $40$ validation grid points chosen randomly. The bottom row concerns  minima (left), mean (middle) and maxima (right) for all $118$ validation grid points. Overall envelopes at the 95\% confidence level are depicted in dark grey.}
        \label{Fig_QQplotsValidation}
    \end{figure}
  \begin{figure}[!h]
    \centering
    \begin{subfigure}[b]{0.49\textwidth}
        \centering
       \includegraphics[width=\textwidth]{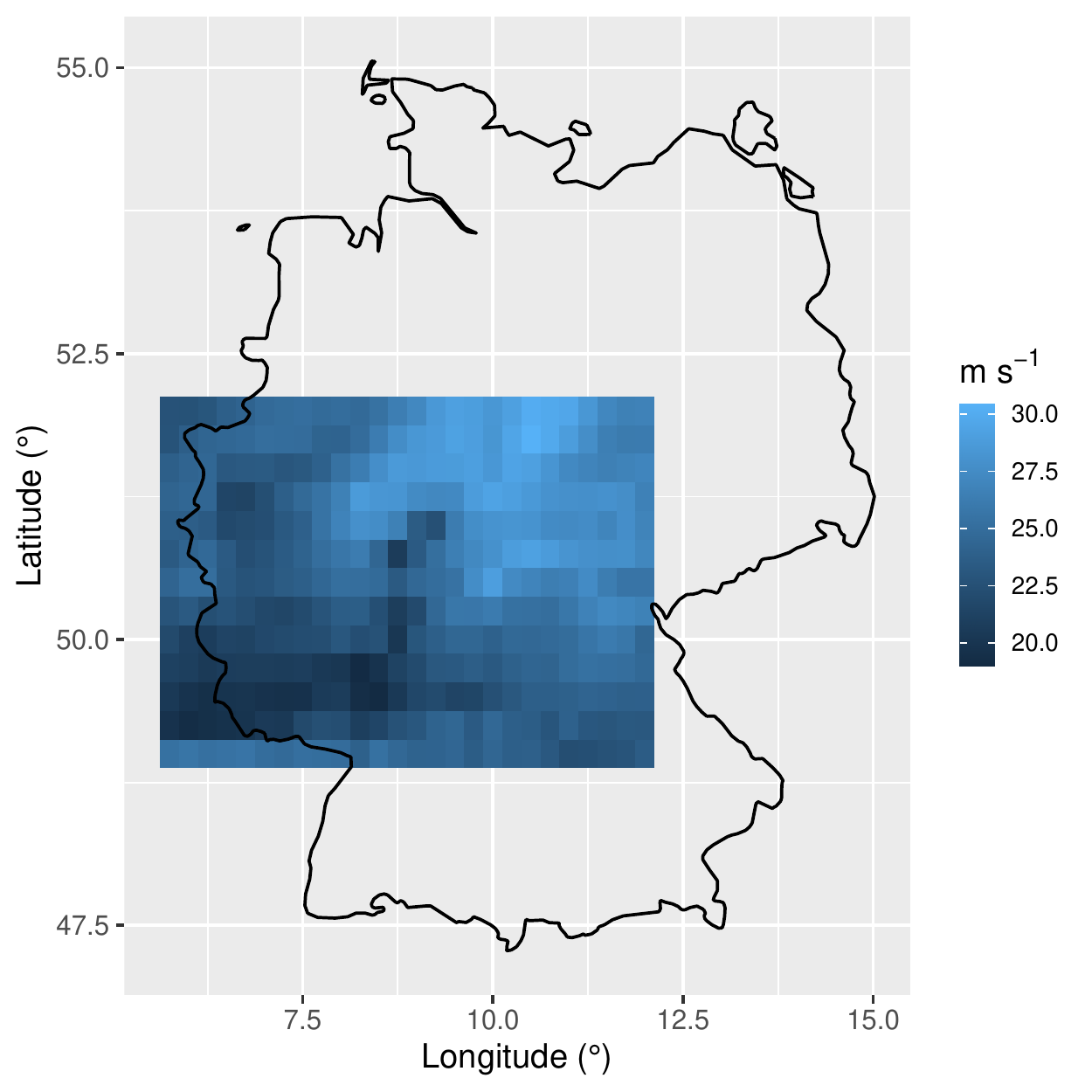}
        \end{subfigure}
        \hfill
        \begin{subfigure}[b]{0.49\textwidth}  
            \centering 
            \includegraphics[width=\textwidth]{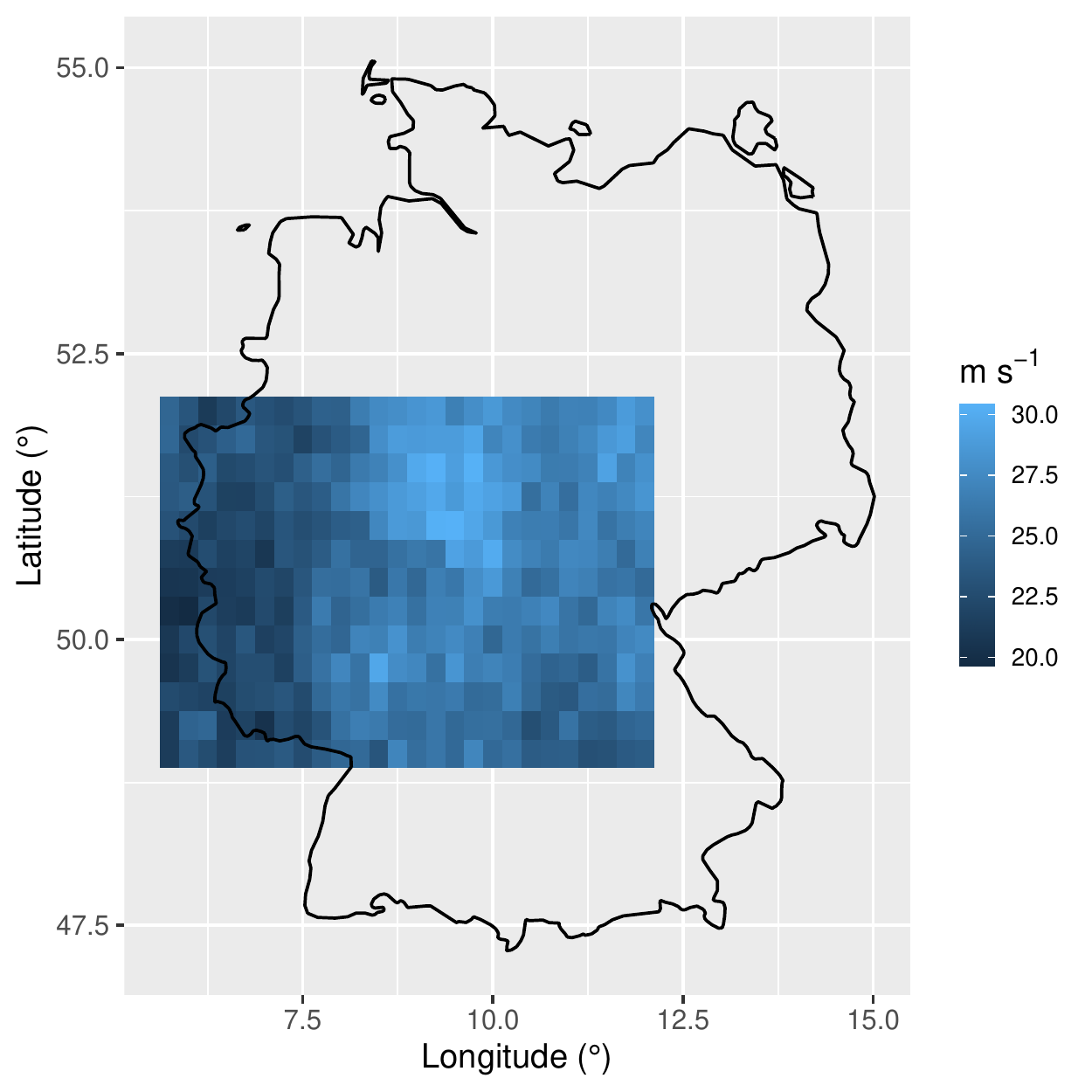}
        \end{subfigure}
        \hfill
        \begin{subfigure}[b]{0.49\textwidth}
        \centering
       \includegraphics[width=\textwidth]{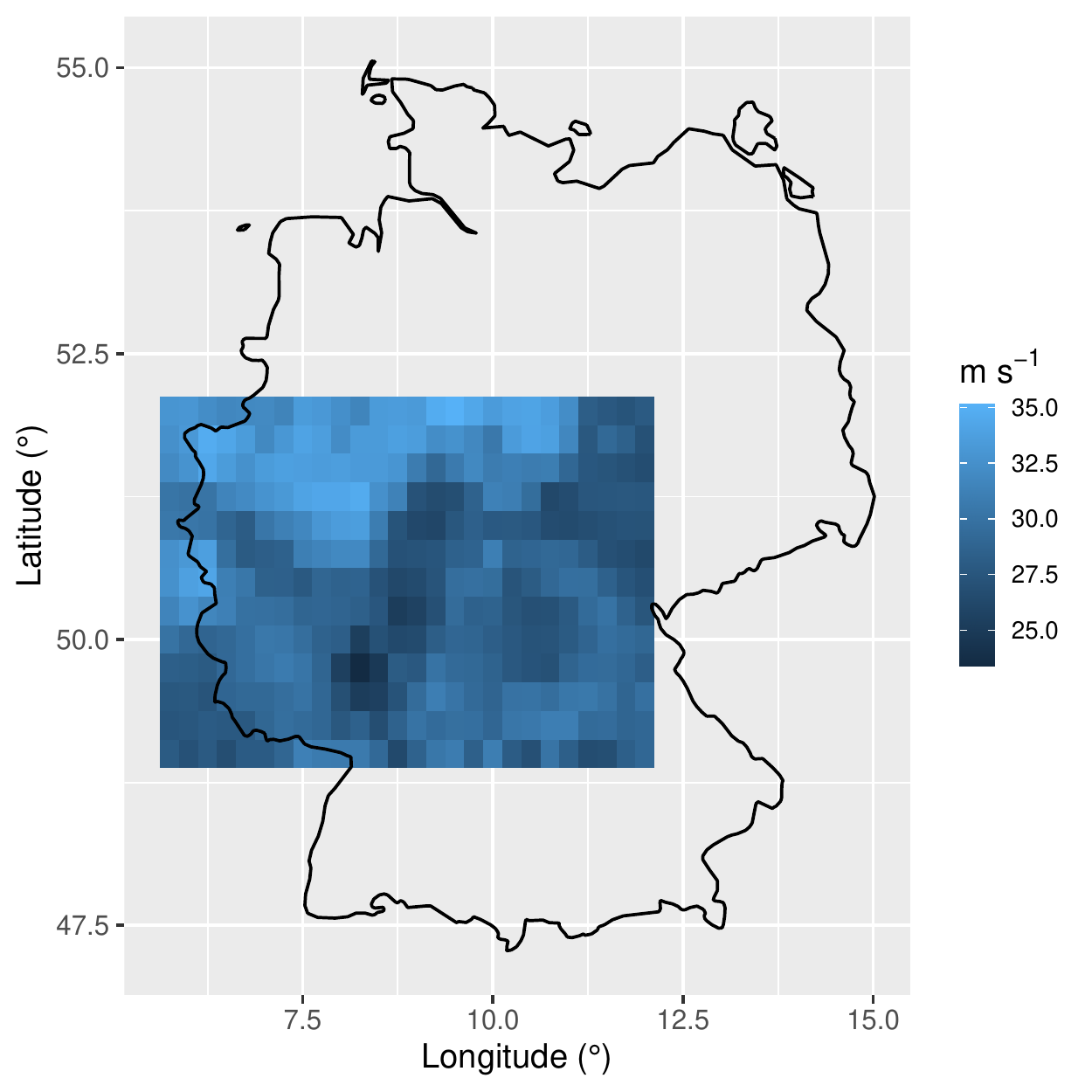}
        \end{subfigure}
        \hfill
        \begin{subfigure}[b]{0.49\textwidth}  
            \centering 
            \includegraphics[width=\textwidth]{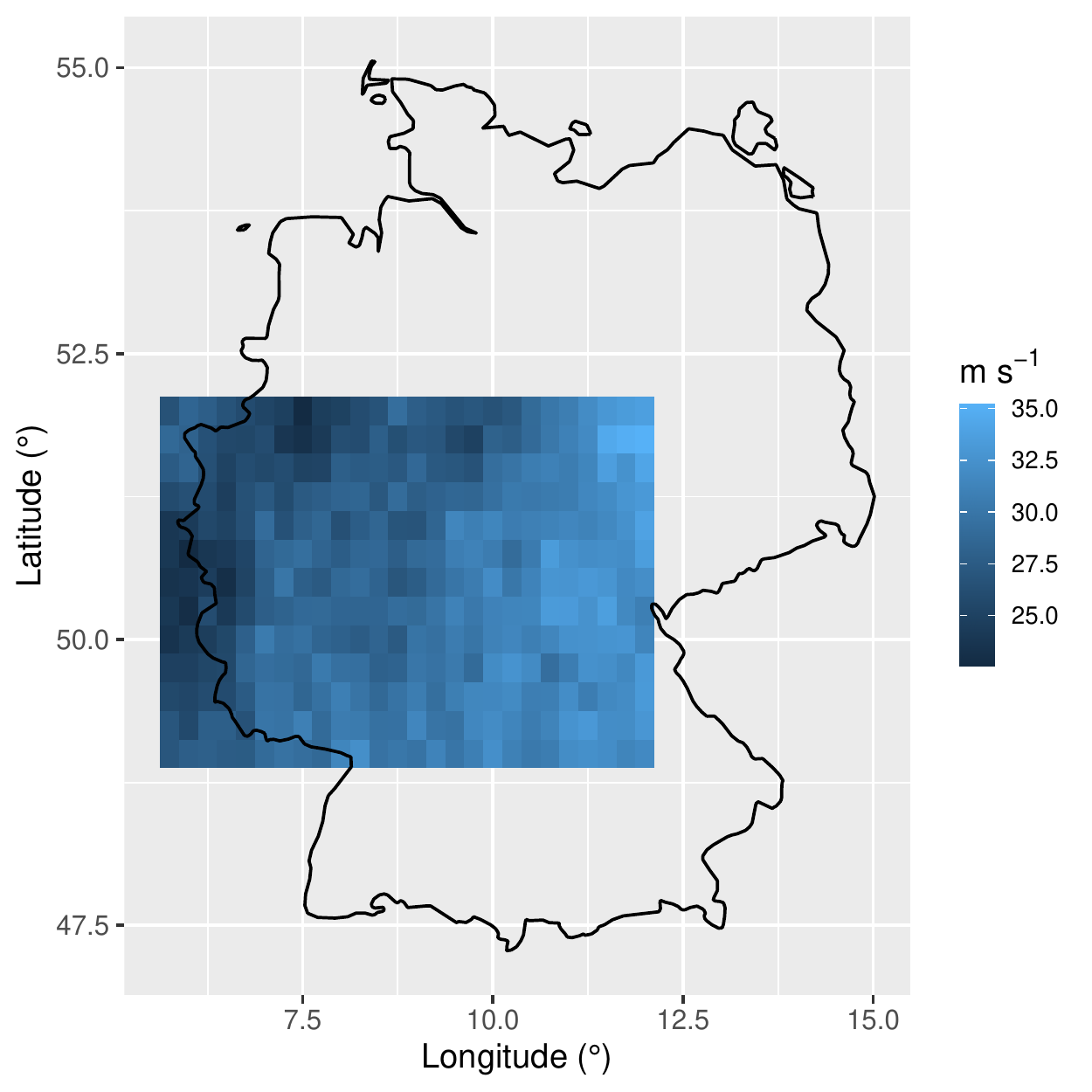}
        \end{subfigure}
        \caption{Comparison between observed fields of pointwise maxima and realizations from the fitted Brown--Resnick model. On the left, pointwise maxima over the period October 2005--March 2006 (top) and the period October 2002--March 2003 (bottom). On the right, examples of realizations from the model having values comparable with those in the first column. 
}
\label{FigPathsRealSimul}
\end{figure}
    
Having shown that our model performs well, we fit it to the data corresponding to all grid points in order to get as accurate parameters' estimates as possible; see Table \ref{Fig_EstimatesAllGridPoints}. Our estimates are in line with general findings on wind speed extremes. Many studies point out that the shape parameter $\xi$ is usually slightly negative, entailing that the distribution of wind speed maxima has a finite right endpoint. E.g., \cite{ceppi2008extreme} obtain a $\xi$ ranging from $-0.2$ to $0$ by fitting a generalized Pareto distribution (GPD) to in situ observations over Switzerland. Similarly, \cite{della2007extreme} fit a GPD to ERA-40 (ECMWF Reanalysis originally intended as a 40-year reanalysis) windstorms data over Europe and find negative values, between $-0.1$ and $-0.3$ on most of land areas; see their Figure 4.15. Typical values for the location and scale parameters $\eta$ and $\tau$ for yearly maxima over Europe are about $25$ m s$^{-1}$ and $3$ m s$^{-1}$, respectively; e.g., considering annual maxima at $35$ weather stations in the Netherlands, \cite{ribatet2013spatial} obtains trend surfaces whose intercepts are about $27$ m s$^{-1}$ for $\eta$ and $3.25$ m s$^{-1}$ 
for $\tau$. Finally, a value of the smoothness parameter $\psi$ between $0.2$ and $1$ seems reasonable; e.g, \cite{ribatet2013spatial} obtains $0.24$ on the Netherlands data and, on similar ones, \cite{einmahl2016m} find $0.40$. We obtain a higher value perhaps because reanalysis data tend to be smoother than in situ observations.
\begin{table}[!h]
\center
\begin{tabular}{c|c|c|c|c}
$\kappa $ & $\psi $ & $\eta $ & $\tau $
& $\xi $ \\ 
\hline
$3.39\ (1.18)$ & $0.81\ (0.05)$ & $25.71\ (0.41)$ & $3.03\
(0.22)$ & $-0.12 \ (0.02)$
\end{tabular}
\newline
\caption{Parameters' estimates (standard errors inside brackets) when using all grid points for the fit.}
\label{Fig_EstimatesAllGridPoints}
\end{table}

\subsection{Results}
\label{Subsec_Results}

Using \eqref{Eq_CostFieldModel}, \eqref{Eq_GeneralDamageFunction} and the facts that $E(\bm{x})>0$ for any $\bm{x} \in \Mbb{R}^2$ and $c_1^{\beta}>0$, we get 
$$ \Mr{Corr}(C(\bm{x}_1), C(\bm{x}_2)) = \Mr{Corr} \left( X^{\beta}(\bm{x}_1), X^{\beta}(\bm{x}_2) \right) = \mathcal{D}_{X, \beta}(\bm{x}_1, \bm{x}_2).$$
Therefore, our dependence measure \eqref{Eq_DepMeas} naturally appears in concrete assessments of the spatial risk associated with extreme wind speed. In this section, we thoroughly study the evolution of $\mathcal{D}_{X, \beta}(\bm{x}_1, \bm{x}_2)$ with respect to $\| \bm{x}_2-\bm{x}_1 \|$, where $X$ is the Brown--Resnick model fitted to the data in Section \ref{Subsubsec_Model}, i.e., with semivariogram \eqref{Eq_Power_Variogram} and parameters in Table \ref{Fig_EstimatesAllGridPoints}, and where $\beta$ has the proper value on our region, i.e., $10$. The integral in $I_{\beta_1, \beta_2}$ (see \eqref{Eq_Def_g_beta1_beta2}) has no closed form and therefore a numerical approximation is required. For this purpose, we use adaptive quadrature with a relative accuracy of $10^{-13}$. Figure \ref{Fig_CorrCurveBeta10} shows a decrease of $\mathcal{D}_{X, \beta}$ from $1$ to $0$ as the Euclidean distance increases, in agreement with our theoretical results of Section \ref{Subsec_TheoreticalContribution}. The decrease is quite slow owing to fairly large range $\kappa$ and  rather low smoothness $\psi$. For two grid points $5 ^\circ$ and $10 ^\circ$ away, $\mathcal{D}_{X, 10}$ is still as high as $0.65$ and $0.48$, respectively. The latter conclusion is however hypothetical as the largest distance between two grid points in our region is about $6.93 ^{\circ}$; fitting our model on a wider region would be possible, but the assumption of spatially-constant GEV parameters and power might be less suitable. This slow decrease points out the need for an insurer to cover a wider region than the one considered here in order to benefit from sufficient spatial diversification.
\begin{figure}[!h]
\begin{center}
\includegraphics[scale=0.7]{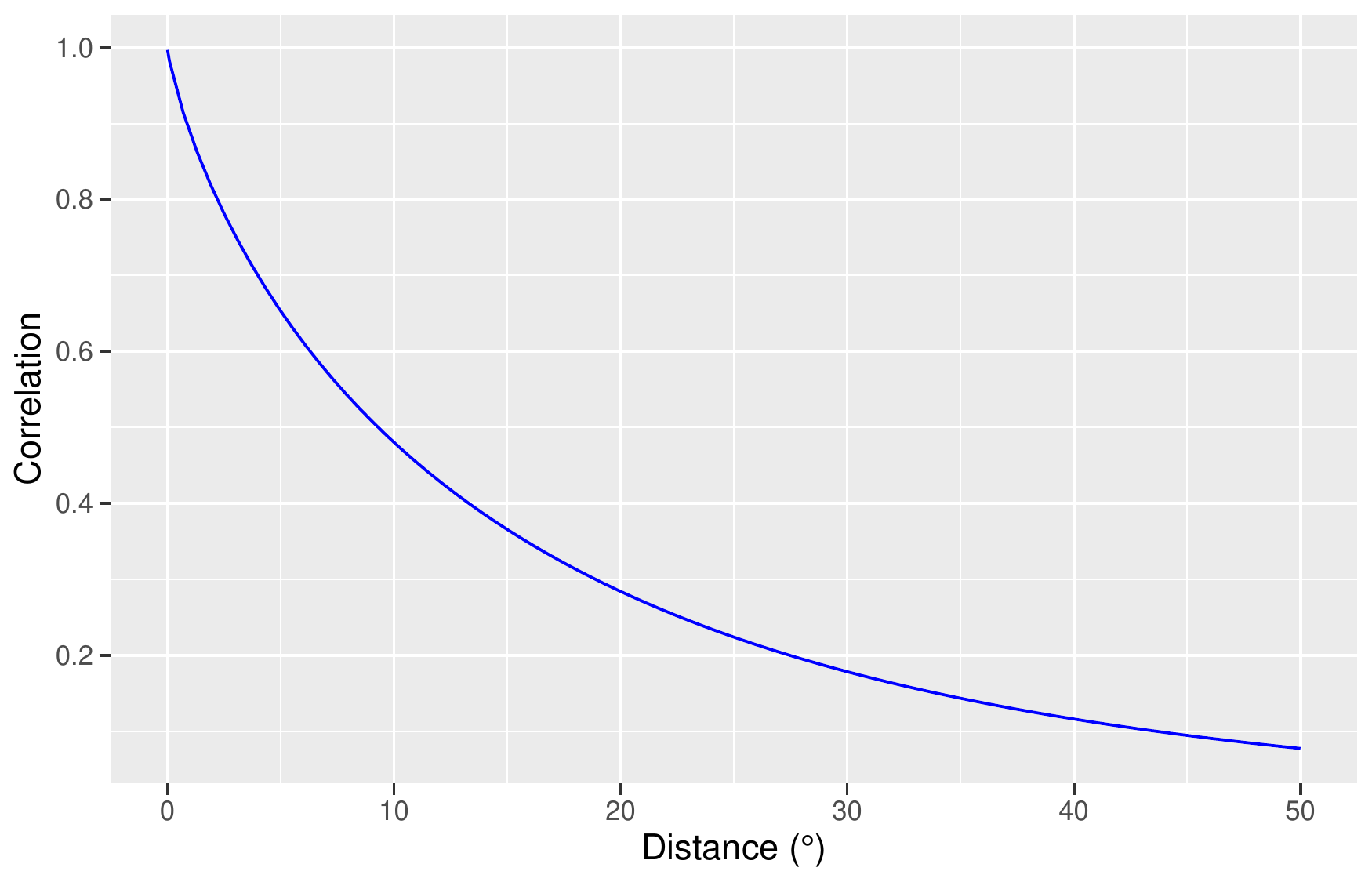}
\caption{Evolution of $\mathcal{D}_{X, 10}(\bm{x}_1, \bm{x}_2)$ with respect to $\| \bm{x}_2-\bm{x}_1 \|$ for $X$ being the Brown--Resnick field with semivariogram \eqref{Eq_Power_Variogram} and parameters in Table \ref{Fig_EstimatesAllGridPoints}.}
\label{Fig_CorrCurveBeta10}
\end{center}
\end{figure}

As already mentioned, various values (basically from $2$ to $12$) of damage powers have been proposed in the literature and the appropriate one may depend on the insurance contract. Moreover, as explained in Section \ref{Sec_Intro}, taking powers (such as the square) of the variables of interest is worthwhile when using correlation as dependence measure. For example, if the true power is $6$, it may also be valuable to study $\Mr{Corr}([X^6(\bm{x}_1)]^2, [X^6(\bm{x}_2)]^2)=\Mr{Corr}(X^{12}(\bm{x}_1), X^{12}(\bm{x}_2))$. For these reasons, investigating how $\mathcal{D}_{X, \beta}(\bm{x}_1, \bm{x}_2)$ varies with $\beta$ for a given max-stable model $X$ and various values of $\bm{x}_1 - \bm{x}_2$ is useful. Figure \ref{Plot_Surface_Corr_4panels} shows that whatever the model considered (including the one fitted to our data) and for any given Euclidean distance, $\mathcal{D}_{X, \beta}$ is only faintly sensitive to the value of $\beta$; more precisely, it very slightly increases in a concave way with $\beta$. On top of being potentially insightful for the understanding of max-stable fields, this finding is valuable for actuarial practice as it shows that making a small error on the evaluation of $\beta$ is not very impactful as far as correlation is concerned. Nonetheless this does not imply that the computations should be done with $\beta=1$ regardless of the true power value. First, although evolving little with $\beta$, our dependence measure is not constant with $\beta$ and so using the right value is recommended for accuracy. Second, $\beta$ strongly affects $\Mr{Var}(X^{\beta}(\bm{x}))$ for any $\bm{x} \in \Mbb{R}^2$, and thus for instance the covariance function and the variance in \eqref{Eq_ExpressionVarianceIntegral}. 

Although the smoothness parameter $\psi$ has been estimated on the data, we also consider various values since $\psi$ heavily affects the rate of decrease of $\mathcal{D}_{X, \beta}$ as the distance between the two sites increase, and thereby the rate of spatial diversification for an insurance company. This allows us to figure out the impact of the use of rougher or smoother data, of estimation error, and of model misspecification. We take $\psi=0.5, 0.81, 1.5, 2$; the value $0.81$ is the one we obtained on our data, $\psi=2$ corresponds to the Smith field with $\Sigma=I_2$ (see \eqref{Eq_Variogram_Smith_Field}), $\psi=1.5$ is intermediate between these two settings, and $\psi=0.5$ corresponds to a quite rough field. In accordance with the discussion at the end of Section \ref{Subsec_TheoreticalContribution}, Figure \ref{Plot_Surface_Corr_4panels} shows that $\mathcal{D}_{X, \beta}$ decreases from $1$ to $0$ as the Euclidean distance increases, and this at a higher rate for larger values of $\psi$. The decrease is faster for the Smith field than for all Brown--Resnick fields having $\psi<2$, and if the true value of $\psi$ is close to $0.5$ or even $0.81$, using the Smith model leads to a serious underestimation of the dependence between insured costs. The minimum Euclidean distance required for $\mathcal{D}_{X, 10}$ to be lower than $0.1$ equals $43.60 ^\circ$ for $\psi=0.81$, instead of around $9.54 ^\circ$ for $\psi=2$ (not shown). 

The results outlined in the two previous paragraphs remain qualitatively unchanged with other values of $\eta$, $\tau$, $\xi$, and choosing a specific value for $\kappa$ does not induce any loss of generality in our study; should $\kappa$ be different, the appropriate plots would be the same as in Figure \ref{Plot_Surface_Corr_4panels} with the values on the x-axis multiplied by the ratio between the true value and the one chosen here. 
\begin{figure}[!h]
\center
\includegraphics[scale=0.85]{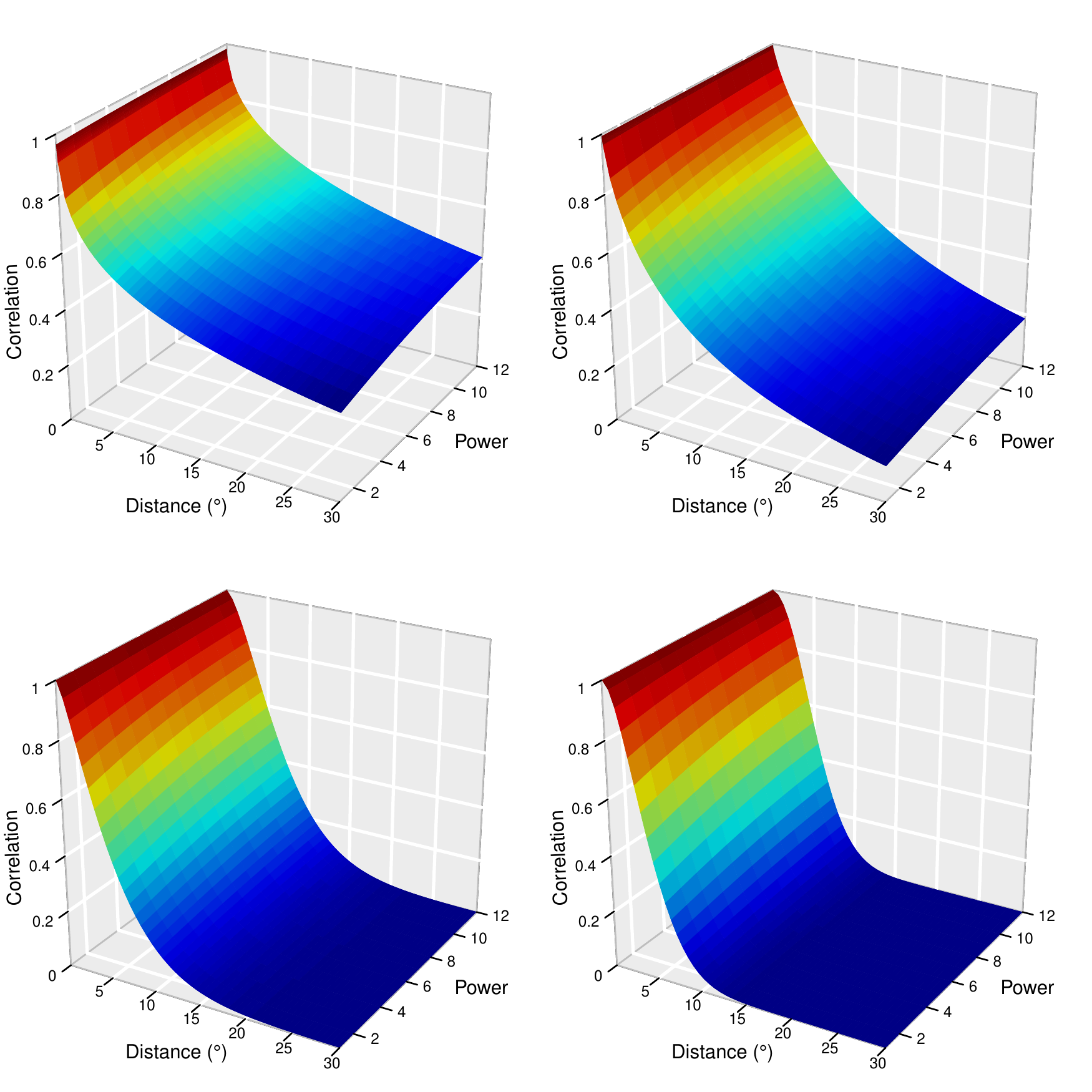}
\caption{Evolution of $\mathcal{D}_{X, \beta}(\bm{x}_2-\bm{x}_1)$ with respect to the distance $\| \bm{x}_2-\bm{x}_1 \|$ and the power $\beta$, where $X$ is the Brown--Resnick field with semivariogram \eqref{Eq_Power_Variogram} with $\psi = 0.5 \mbox{ (top left), } 0.81 \mbox{ (top right), } 1.5 \mbox{ (bottom left) and } 2 \mbox{ (bottom right)}$, and whose other parameters are given in Table \ref{Fig_EstimatesAllGridPoints}.}
\label{Plot_Surface_Corr_4panels}
\end{figure}

Finally we briefly study the extension of \eqref{Eq_DepMeas} where the marginal parameters and the power  are site-specific. We consider two sites $\bm{x}_1, \bm{x}_2$ that are $3 ^{\circ}$ away, but our findings hold more generally. We successively investigate the effects of a spatially-varying power, location, scale and shape; more precisely we evaluate \eqref{Eq_CorrVaryingBetaEtaTauXi} where $X$ is the Brown--Resnick model with semivariogram \eqref{Eq_Power_Variogram}
\Bit
\item with parameters in Table \ref{Fig_EstimatesAllGridPoints} and $\beta(\bm{x}_1), \beta(\bm{x}_2) \in \{ 1, \ldots, 12\}$.
\item with parameters in Table \ref{Fig_EstimatesAllGridPoints} apart from the location ($\eta(\bm{x}_1), \eta(\bm{x}_2) \in [15, 35]$), and $\beta(\bm{x}_1)=\beta(\bm{x}_2)=10$.
\item with parameters in Table \ref{Fig_EstimatesAllGridPoints} apart from the scale ($\tau(\bm{x}_1), \tau(\bm{x}_2) \in [2, 4]$), and $\beta(\bm{x}_1)=\beta(\bm{x}_2)=10$.  
\item with parameters in Table \ref{Fig_EstimatesAllGridPoints} apart from the shape ($\xi(\bm{x}_1), \xi(\bm{x}_2) \in [-0.2, -0.06]$), and $\beta(\bm{x}_1)=\beta(\bm{x}_2)=10$.  
\Eit
The ranges for the GEV parameters have been chosen to be approximately centred on the estimates obtained on the data. Figure \ref{FigCorrelationCombinations} shows that, for a fixed $\beta(\bm{x}_1)$, the correlation increases with $\beta(\bm{x}_2)$ on $[1, \beta(\bm{x}_1)]$ and then decreases. The highest correlation is thus obtained for $\beta(\bm{x}_2) = \beta(\bm{x}_1) = \beta$, and, as already seen, slightly increases in a concave way when $\beta$ increases. Also, the higher the difference between $\beta(\bm{x}_1)$ and $\beta(\bm{x}_2)$, the lower the correlation. Similar conclusions hold for the scale and shape parameters, although the variations of the correlation are smaller for the chosen range of values. For $\tau(\bm{x}_1)=\tau(\bm{x}_2)=\tau$, the increase with respect to $\tau$ is concave, whereas for $\xi(\bm{x}_1)=\xi(\bm{x}_2)=\xi$, the increase with respect to $\xi$ is linear. The findings for the location are similar to those for the scale and shape although, for  $\eta(\bm{x}_1)=\eta(\bm{x}_2)=\eta$, the correlation slowly decreases in a concave way as $\eta$ increases.
\begin{figure}[!h]
    \centering
    \begin{subfigure}[b]{0.49\textwidth}
        \centering
       \includegraphics[width=\textwidth]{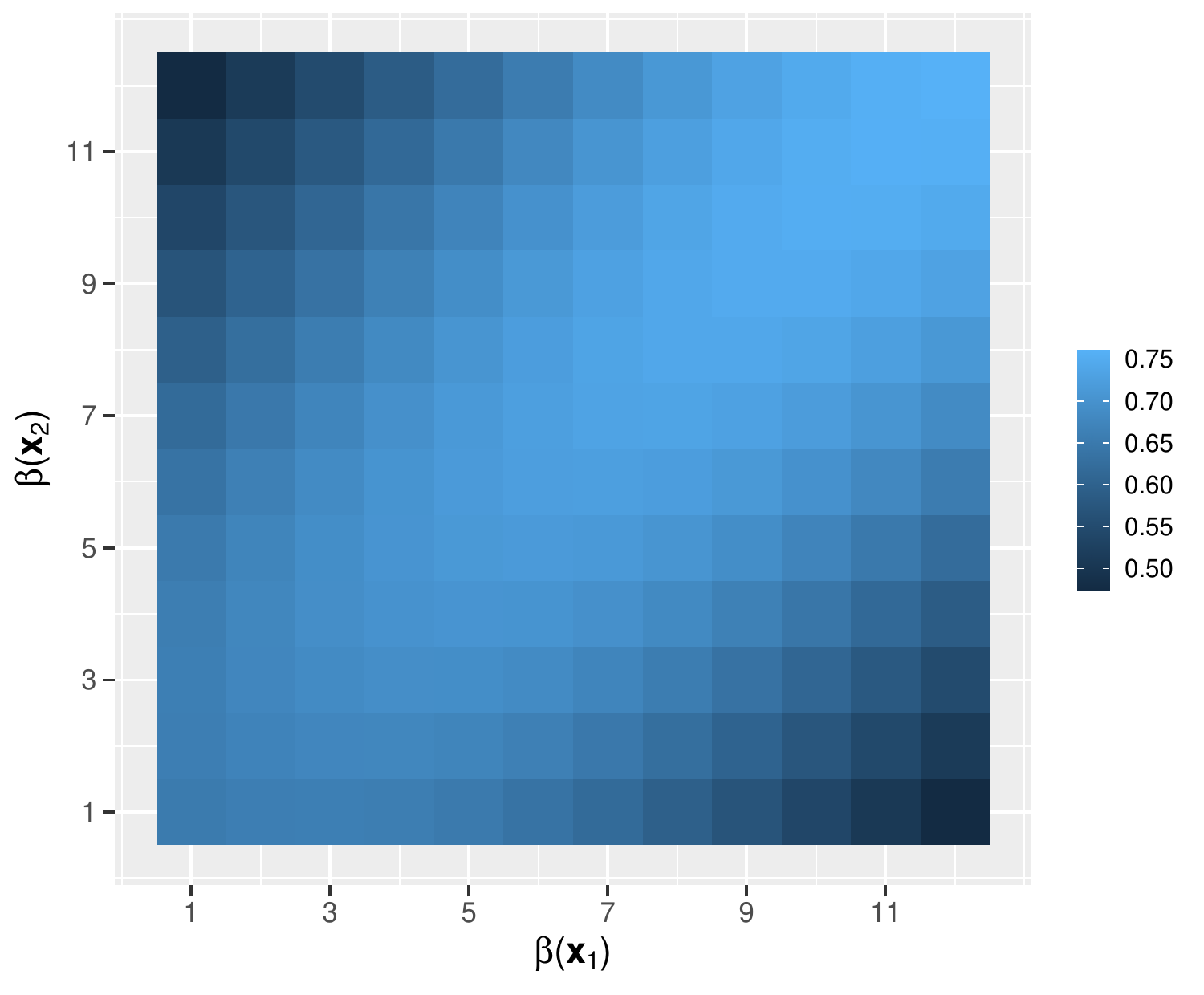}
        \end{subfigure}
        \hfill
        \begin{subfigure}[b]{0.49\textwidth}  
            \centering 
            \includegraphics[width=\textwidth]{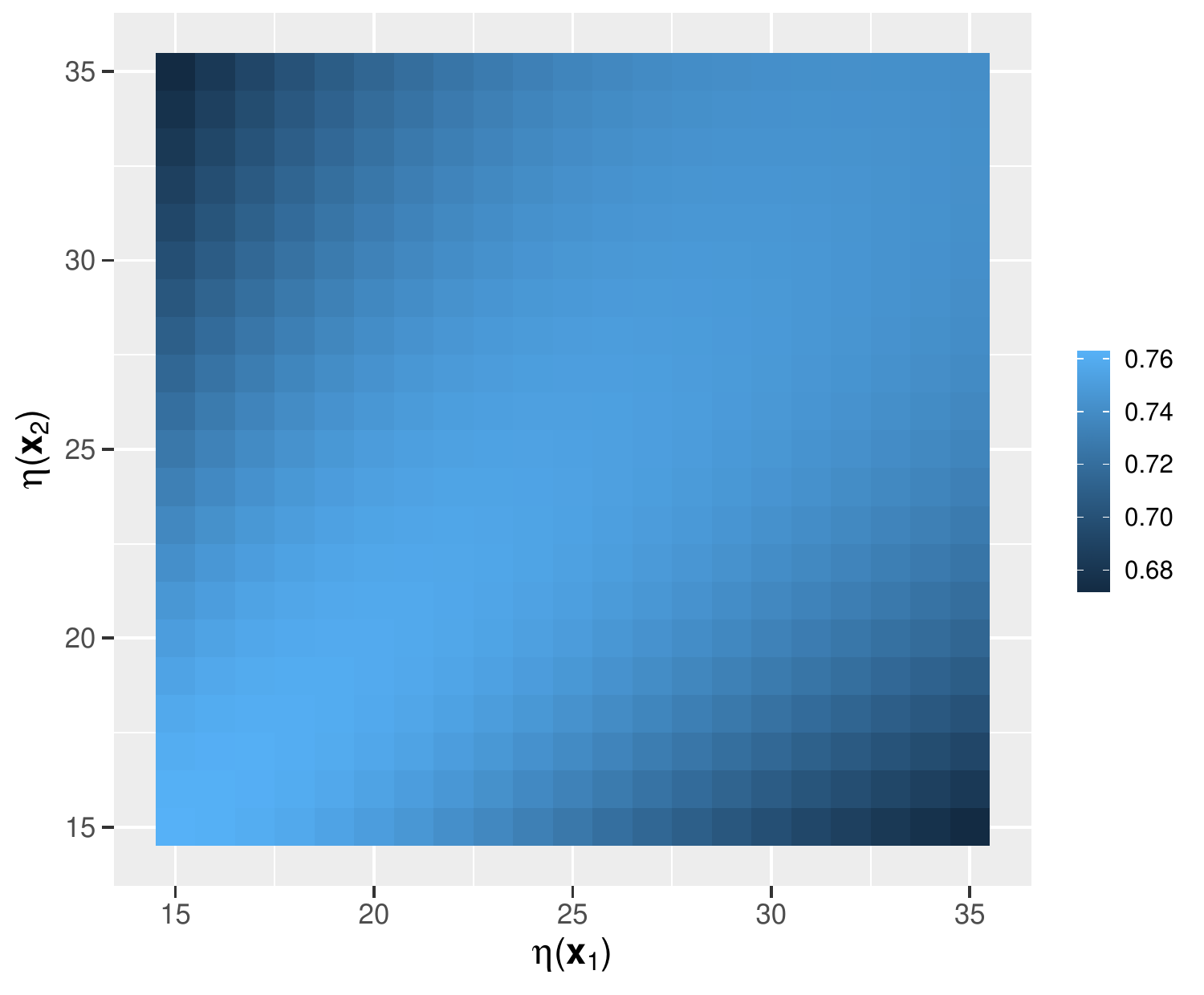}
        \end{subfigure}
        \hfill
        \begin{subfigure}[b]{0.49\textwidth}
        \centering
       \includegraphics[width=\textwidth]{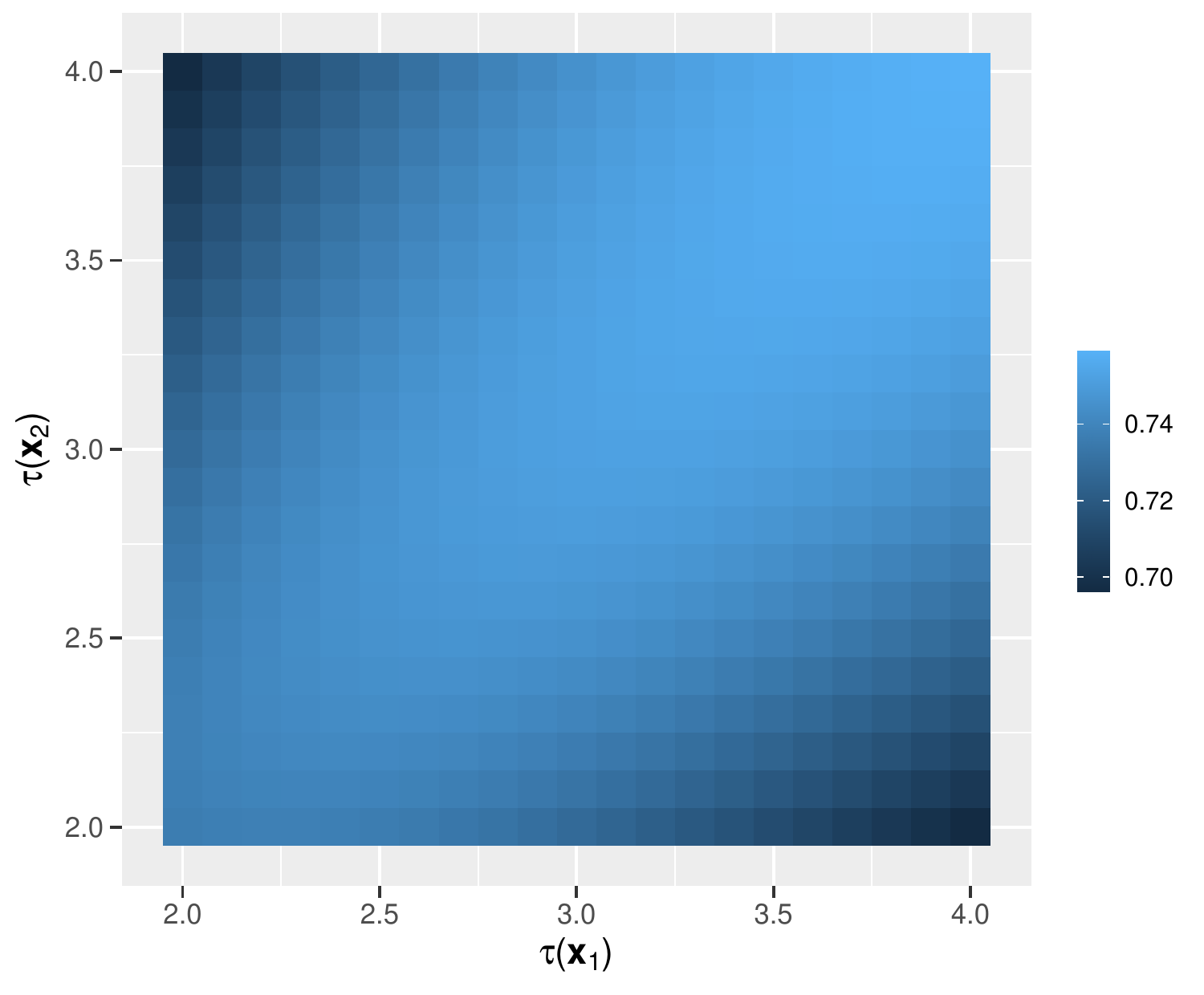}
        \end{subfigure}
        \hfill
        \begin{subfigure}[b]{0.49\textwidth}  
            \centering 
            \includegraphics[width=\textwidth]{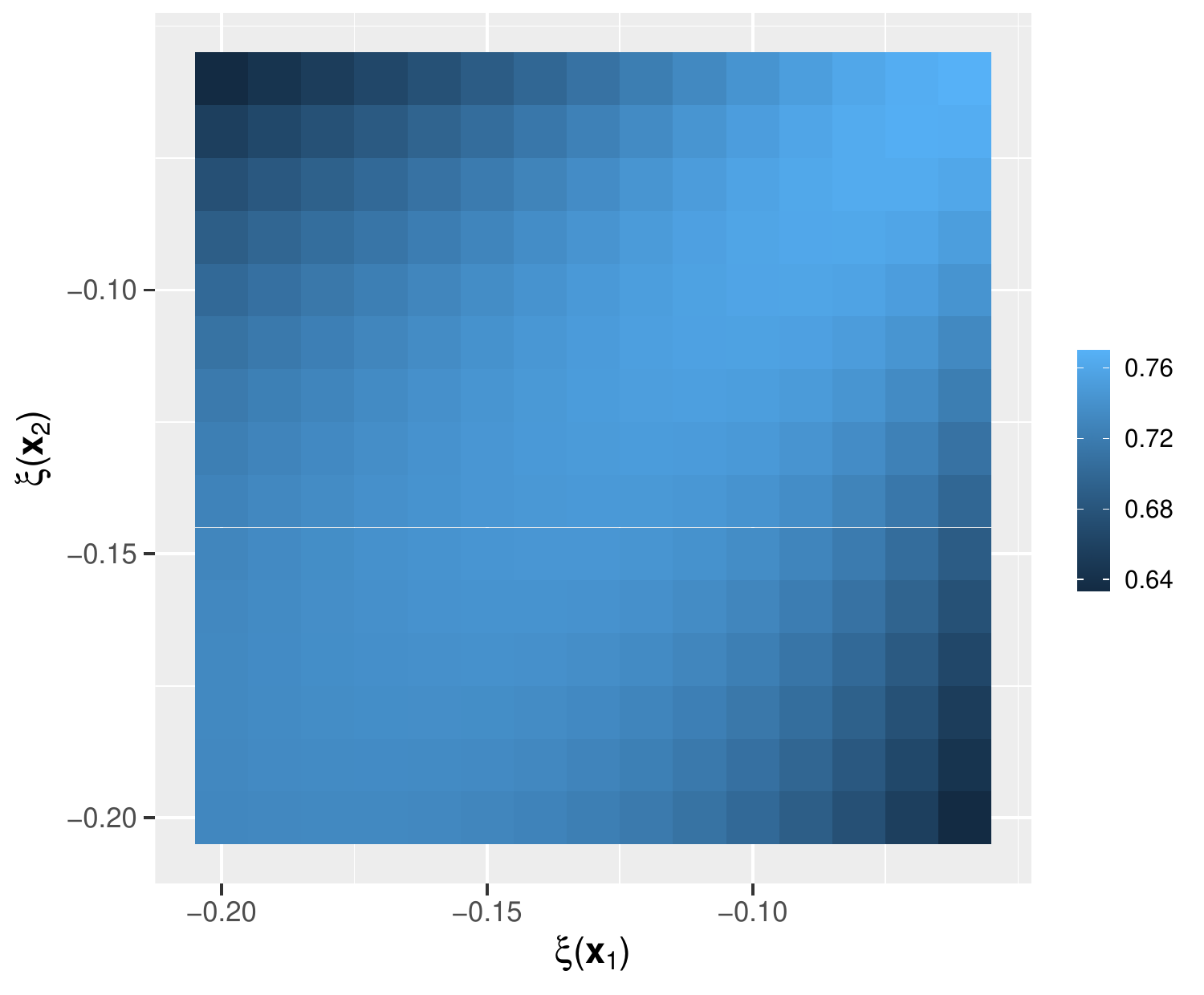}
        \end{subfigure}
        \caption{Heatmap of $\Mr{Corr} ( X^{\beta(\bm{x}_1)}(\bm{x}_1), X^{\beta(\bm{x}_2)}(\bm{x}_2))$, where $X$ is the Brown--Resnick field with semivariogram \eqref{Eq_Power_Variogram} with: parameters in Table \ref{Fig_EstimatesAllGridPoints} and $\beta(\bm{x}_1), \beta(\bm{x}_2) \in \{ 1, \ldots, 12\}$ (top left); parameters in Table \ref{Fig_EstimatesAllGridPoints} apart from the location ($\eta(\bm{x}_1), \eta(\bm{x}_2) \in [15, 35]$), and $\beta(\bm{x}_1)=\beta(\bm{x}_2)=10$ (top right); parameters in Table \ref{Fig_EstimatesAllGridPoints} apart from the scale ($\tau(\bm{x}_1), \tau(\bm{x}_2) \in [2, 4]$), and $\beta(\bm{x}_1)=\beta(\bm{x}_2)=10$ (bottom left); parameters in Table \ref{Fig_EstimatesAllGridPoints} apart from the shape ($\xi(\bm{x}_1), \xi(\bm{x}_2) \in [-0.2, -0.06]$), and $\beta(\bm{x}_1)=\beta(\bm{x}_2)=10$.
        }
\label{FigCorrelationCombinations}
\end{figure}

\section{Conclusion}
\label{Sec_Conclusion}

Hüsler--Reiss vectors and Brown--Resnick fields are popular and widely used models for componentwise and pointwise maxima. We provide explicit formulas for the correlation between powers of the components of bivariate Hüsler--Reiss vectors and deduce analytical expressions for the correlation function of powers of Brown--Resnick fields. Although extremal models are considered, studying the correlation function makes sense as the latter is required when we are interested in the variance or the asymptotic distribution of the spatial integral of a field, which is typically the case in spatial risk assessment. The application of a power transform to random variables before taking the correlation allows detection of part of non-linear dependence and is therefore common practice in financial time series analysis. Moreover, the relevance of powers as damage functions for natural disasters is largely documented in the literature. In the second part of the paper, we use our theoretical contributions and reanalysis wind gust data to study the spatial dependence of modelled insured losses from extreme wind speeds for residential buildings in Germany. We find that the dependence decreases slowly with the distance and that our dependence measure is not very sensitive to the power value.

The theoretical results obtained here are used in \cite{KochStochDerEstim2022} as well as in an ongoing study where spatial risk measures \citep{koch2017spatial, koch2019SpatialRiskAxioms} are applied to concrete assessment of the risk of impacts from extreme wind speeds. Other potentially interesting applications of the derived expressions are flood risk assessment and moment-based estimation of the parameters of Hüsler--Reiss vectors or Brown--Resnick fields. A more detailed study, both theoretically and numerically, of the correlation function expressed in Remark \ref{RqCorrBRDiffMarPow} (non-stationary case) would be welcome, and deriving analytical formulas of \eqref{Eq_DepMeas} for other classes of max-stable fields such as the extremal $t$ model \citep{opitz2013extremal} as well as $r$-Pareto fields \citep[e.g.,][]{de2018high} would be useful for applications.

Although our insured loss model is supported by the literature, thoroughly assessing its performance on insured loss data is prominent for practice, and this is done in an ongoing work. Finally, in the case study we take the value of $\beta$ obtained in other papers as given, and an approach that would involve estimating $\beta$ as well would consist in fitting powers of rescaled max-stable fields to insured loss data directly.

\section*{Acknowledgements}

The author wishes to thank Christian
Y. Robert for theoretical discussions, John Ery for exchanges about reanalysis data as well as Anthony C. Davison and Thomas Mikosch for some comments. He also would like to acknowledge the Swiss National Science Foundation (project
200021\_178824) and the Institute of Mathematics at EPFL for financial support.

\newpage

\appendix

\section{Proofs}

\subsection{For Theorem \ref{Th_CovarianceSimpleHR}}

\begin{proof}
First, we show the result for $h=0$. In that case, $Z_1=Z_2$ a.s. \citep[e.g.,][Section 2]{husler1989maxima}. Hence, since $Z_1$ and $Z_2$ follow the standard Fr\'echet distribution, 
$\mathbb{E} [Z_i^{\beta_i}]= \Gamma(1-\beta_i)$, $i=1, 2$, and thus
\begin{align*}
\mathrm{Cov} \left( Z_1^{\beta_1}, Z_2^{\beta_2} \right) =\Gamma(1-\beta_1-\beta_2)-\Gamma(1-\beta_1)\Gamma(1-\beta_2) = I_{\beta_1, \beta_2}(0)-\Gamma(1-\beta_1)\Gamma(1-\beta_2).
\end{align*} 

Now, we prove the result for $h>0$. We have
$$\mathbb{E} \left[ Z_1^{\beta_1} Z_2^{\beta_2} \right] = \displaystyle \int_{0}^{\infty} \displaystyle \int_{0}^{\infty} z_1^{\beta_1} z_2^{\beta_2} l(z_1, z_2) \mathrm{d}z_1 \mathrm{d}z_2,$$
where $l$ denotes the bivariate density of $\bm{Z}$.
In order to take advantage of the radius/angle decomposition of multivariate extreme-value distributions, we make the change of variable
$$
\begin{pmatrix}
z_1 \\
z_2
\end{pmatrix}
=\begin{pmatrix}
u \\
\theta\  u
\end{pmatrix}
=\begin{pmatrix}
\Psi_1(u, \theta) \\
\Psi_2(u, \theta)
\end{pmatrix}
=\Psi(u, \theta).
$$
The corresponding Jacobian matrix is written
$$
J_{\Psi}(u, \theta)=
\begin{pmatrix}
1 & 0 \\
& \\
\theta & u
\end{pmatrix},
$$  
and its determinant is thus $\det(J_{\Psi}(u, \theta))=u$.
Therefore, introducing 
$$a(z_1,z_2)=z_1^{\beta_1} z_2^{\beta_2} l(z_1, z_2), \quad z_1, z_2 >0,$$ 
we have
\begin{align}
\mathbb{E} \left[ Z_1^{\beta_1} Z_2^{\beta_2} \right] & = \int_{0}^{\infty} \int_{0}^{\infty} a(z_1,z_2) \mathrm{d}z_1 \mathrm{d}z_2 \nonumber
\\& = \int \int_{\Psi^{-1}((0, \infty )^2)} a(\Psi(u, \theta)) \det(J_{\Psi}(u, \theta)) \mathrm{d}u\mathrm{d}\theta \nonumber
\\&  = \int_{0}^{\infty} \int_{0}^{\infty} u^{\beta_1} \theta^{\beta_2} u^{\beta_2} l(u, \theta u) u \mathrm{d}u \mathrm{d}\theta \nonumber
\\&  = \int_{0}^{\infty} \int_{0}^{\infty} u^{\beta_1+\beta_2 +1 } \theta^{\beta_2} l(u, \theta u) \mathrm{d}u \mathrm{d}\theta.
\label{Chapter_Riskmeasures_Eq1_Exp_Damage}
\end{align}
Differentiation of \eqref{Eq_HuslerReissDist} yields \cite[see, e.g.,][Equation (4)]{padoan2010likelihood}, for $z_1, z_2>0$,
\begin{align}
l(z_1, z_2) =\exp \left( -\frac{\Phi(w)}{z_1}-\frac{\Phi(v)}{z_2} \right) \times \bigg[ \left( \frac{\Phi(w)}{z_1^2}+\frac{\phi(w)}{h z_1^2}-\frac{\phi(v)}{h z_1 z_2} \right)  &\times \left( \frac{\Phi(v)}{z_2^2}+\frac{\phi(v)}{h z_2^2}-\frac{\phi(w)}{h z_1 z_2} \right) \nonumber
 \\&  +\left( \frac{v \phi(w)}{h^2 z_1^2 z_2}+\frac{w \phi(v)}{h^2 z_1 z_2^2} \right) \bigg],
\label{Chapter_RiskMeasure_Smith_Bivariate_Density}
\end{align}
where 
$$w=\frac{h}{2}+\frac{\log \left( z_2 / z_1 \right)}{h}
\quad \mbox{and} \quad v=\frac{h}{2}-\frac{\log \left( z_2/z_1 \right)}{h}.$$

Therefore, for any $u, \theta>0$,
\begin{align}
& \quad \ l(u, \theta u) \nonumber \\
& = \exp \left( - \frac{1}{u} \left[ \Phi \left( \frac{h}{2}+ \frac{\log \theta}{h} \right)+\frac{1}{\theta} \Phi \left( \frac{h}{2} - \frac{\log \theta}{h} \right) \right] \right)
\times  \bigg \{ \frac{1}{u^4} \bigg[   \Phi \left( \frac{h}{2}+\frac{\log \theta}{h} \right) +\frac{1}{h} \phi \left( \frac{h}{2}+ \frac{\log \theta}{h} \right) \nonumber
\\& \ \ \  -\frac{1}{h \theta}\phi \left(  \frac{h}{2}-\frac{\log \theta}{h} \right) \bigg]
\times \left[ \frac{1}{\theta^2} \Phi \left(  \frac{h}{2}- \frac{\log \theta}{h} \right)+\frac{1}{h \theta^2} \phi \left( \frac{h}{2}- \frac{\log \theta}{h} \right)-\frac{1}{h \theta} \phi \left( \frac{h}{2}+ \frac{\log \theta}{h} \right) \right] \nonumber
\\& \ \ \  + \frac{1}{u^3} \left[ \frac{1}{h^2 \theta} \left( \frac{h}{2}- \frac{\log \theta}{h} \right) \ \phi \left( \frac{h}{2}+ \frac{\log \theta}{h} \right)+\frac{1}{h^2 \theta^2} \left( \frac{h}{2}+ \frac{\log \theta}{h} \right) \ \phi \left( \frac{h}{2}-\frac{\log \theta}{h}  \right) \right] \bigg \} \nonumber
\\& =  \exp \left( - \frac{C_1(\theta,h)}{u} \right) \left( \frac{C_2(\theta,h)}{u^4} + \frac{C_3(\theta,h)}{u^3} \right).
\label{Eq_h_u_tu}
\end{align}
\Tb{We denote by $\mathcal{F}_{s_f}$ the Fr\'echet distribution with shape and scale parameters $1$ and $s_f>0$, i.e., if $X \sim \mathcal{F}_{s_f}$,
$\mathbb{P}(X \leq x)= \exp ( - s_f/x ), x > 0.$
Using \eqref{Chapter_Riskmeasures_Eq1_Exp_Damage} and \eqref{Eq_h_u_tu} and the fact that the density of $X \sim \mathcal{F}_{s_f}$ is
$l_f(x)=s_f/x^2 \exp \left( -s_f/x \right)$}, we obtain
\begin{align}
&\quad \ \mathbb{E} \left[ Z_1^{\beta_1} Z_2^{\beta_2} \right] \nonumber \\
&= \int_{0}^{\infty} \theta^{\beta_2} \left( \int_{0}^{\infty}  u^{\beta_1+\beta_2 +1} \exp \left( - \frac{C_1(\theta,h)}{u} \right) \left( \frac{C_2(\theta,h)}{u^4} + \frac{C_3(\theta,h)}{u^3} \right) \mathrm{d}u \right) \mathrm{d}\theta \nonumber
\\& = \int_{0}^{\infty} C_2(\theta,h) \ \theta^{\beta_2} \left( \int_{0}^{\infty}  u^{\beta_1+\beta_2 -3 } \exp \left( - \frac{C_1(\theta,h)}{u} \right) \mathrm{d}u \right) \mathrm{d}\theta \nonumber
\\& \ \ \  + \int_{0}^{\infty} C_3(\theta,h) \ \theta^{\beta_2} \left( \int_{0}^{\infty}  u^{\beta_1+\beta_2 -2 } \exp \left( - \frac{C_1(\theta,h)}{u} \right) \mathrm{d}u \right) \mathrm{d}\theta \nonumber
\\& = \int_{0}^{\infty} C_2(\theta,h) \ \theta^{\beta_2} \left( \int_{0}^{\infty}  u^{\beta_1+\beta_2 -1 }\  \frac{1}{u^2}\  \exp \left( - \frac{C_1(\theta,h)}{u} \right) \mathrm{d}u \right) \mathrm{d}\theta \nonumber
\\& \ \ \ + \int_{0}^{\infty} C_3(\theta,h) \ \theta^{\beta_2} \left( \int_{0}^{\infty}  u^{\beta_1+\beta_2} \ \frac{1}{u^2} \exp \left( - \frac{C_1(\theta,h)}{u} \right) \mathrm{d}u \right) \mathrm{d}\theta \nonumber
\\& = \int_{0}^{\infty} \frac{C_2(\theta,h)}{C_1(\theta,h)} \ \theta^{\beta_2} \ \mu_{ \beta_1 + \beta_2 -1} \left( \mathcal{F}_{C_1(\theta,h)} \right)\mathrm{d}\theta + \int_{0}^{\infty} \frac{C_3(\theta,h)}{C_1(\theta,h)} \ \theta^{\beta_2} \ \mu_{ \beta_1 + \beta_2} \left( \mathcal{F}_{C_1(\theta,h)} \right) \mathrm{d}\theta,
\label{Eq_Expectation}
\end{align}
where $\mu_k(F)$ stands for the $k$-th moment of a random variable having $F$ as distribution. It is immediate to see that $\mu_k ( \mathcal{F}_{s_f} ) =s_f^k \  \Gamma(1-k)$,
which, combined with \eqref{Eq_Expectation}, yields the result.
\end{proof}

\subsection{For Theorem \ref{Th_Cov_Maxstab_Real_Marg}}

\begin{proof}
Using \eqref{Eq_LinkMaxstbVect_SimpleMaxstabVect} and the binomial theorem, we obtain
\begin{align*}
\Mr{Cov} \left( X_1^{\beta_1}, X_2^{\beta_2} \right) = \sum_{k_1=0}^{\beta_1} \sum_{k_2=0}^{\beta_2} & {\beta_1 \choose k_1} \left( \eta_1-\frac{\tau_1}{\xi_1} \right)^{k_1} \left( \frac{\tau_1}{\xi_1} \right)^{\beta_1-k_1} {\beta_2 \choose k_2} \left( \eta_2-\frac{\tau_2}{\xi_2} \right)^{k_2} \left( \frac{\tau_2}{\xi_2} \right)^{\beta_2-k_2} \nonumber \\
& \times \Mr{Cov}  \left(Z_1^{(\beta_1-k_1) \xi_1},  Z_2^{(\beta_2-k_2) \xi_2} \right),
\end{align*}
which directly yields \eqref{EqCovHRDiffBetaMar} by Theorem \ref{Th_CovarianceSimpleHR}.

If $Z$ is standard Fréchet, $\Mbb{E}(Z^{\beta^*})=\Gamma(1-\beta^*)$ for any $\beta^* < 1/2$, which gives, for $\beta_1^*, \beta_2^* < 1/2$, 
$$\Mr{Cov}(Z^{\beta_1^*}, Z^{\beta_2^*})= \Gamma(1 - [\beta_1^* + \beta_2^*]) - \Gamma(1 - \beta_1^*) \Gamma(1 - \beta_2^*).$$
Using this together with \eqref{Eq_LinkMaxstbVect_SimpleMaxstabVect} and the binomial theorem yields \eqref{Eq_Var_GEV_beta_i}.
\end{proof}

\subsection{For Proposition \ref{Prop_Continuity_Cov_xi_0}}

\begin{proof}
For $i=1,2$, $X_{i, \xi}$ follows the GEV distribution with parameters $\eta_i$, $\tau_i$ and $\xi$, the density of which we denote by $f_i$.
Let us assume that $\xi \in S_{\beta_1, \beta_2, \varepsilon}$ and $\xi>0$. We have for all $\alpha>0$
\Beq
\label{Eq_ExpressionMomentalpha}
\mathbb{E}[| X_{i, \xi}|^{\alpha}] = \int_{\eta_i-\tau_i/\xi}^{0} |x|^{\alpha} f_i(x) \mathrm{d}x + \int_{0}^{\infty} x^{\alpha} f_i(x) \mathrm{d}x
\Eeq
and thus
\Beq
\label{Eq_Inequality_SupMomentalpha}
\sup_{\xi \in \mathcal{S}} \mathbb{E}[| X_{i, \xi}|^{\alpha}] \leq \sup_{\xi \in \mathcal{S}}  \int_{\eta_i-\tau_i/\xi}^{0} |x|^{\alpha} f_i(x) \mathrm{d}x + \sup_{\xi \in \mathcal{S}} \int_{0}^{\infty} x^{\alpha} f_i(x) \mathrm{d}x\Eeq
for any subset $\mathcal{S}$ of $(0, \infty)$.
We deal with the second integral in \eqref{Eq_ExpressionMomentalpha}, for which there is a potential problem at $\infty$. We have
\begin{align}
& \quad \int_{0}^{\infty} x^{\alpha} \exp \left( -[1+\xi (x-\eta_i)/\tau_i]^{-1/\xi} \right) [1+\xi (x-\eta_i)/\tau_i]^{-1/\xi-1} \mathrm{d}x \nonumber \\&=\int_{0}^{1} \left[ \eta_i + \tau_i(z^{-\xi}-1)/\xi \right]^{\alpha} \exp(-z) \mathrm{d}z,
\label{Eq_SndIntegralVariableChanged}
\end{align}
where we used the change of variable $z=[1+\xi (x-\eta_i)/\tau_i]^{-1/\xi}$.
As $[ \eta_i + \tau_i(z^{-\xi}-1)/\xi ] \underset{z \to 0}{\sim} \tau_i z^{-\xi}/\xi$, \eqref{Eq_SndIntegralVariableChanged} is finite provided $\alpha \xi <1$. Choose $0<\xi^*<1/\alpha$, such that \eqref{Eq_SndIntegralVariableChanged} computed at $\xi^*$ is finite. Introducing $g(\xi)= (z^{-\xi}-1)/\xi$, $\xi>0$, where $z\geq 0$, we have
$$\frac{\mathrm{d}g(\xi)}{\Mr{d}\xi}=\frac{z^{-\xi}(\log(z^{-\xi}) - 1)}{\xi^2} + \frac{1}{\xi^2}.$$
A well-known inequality states that $\log(z^{-\xi})\geq 1 - 1 / z^{-\xi}$ for any $z\geq0$, which yields $z^{-\xi}(\log(z^{-\xi}) - 1) \geq -1$ and thus $g'(\xi)\geq0$. Combined with the fact that $0 \leq z \leq 1$, this gives for any $0 < \xi \leq \xi^*$
\begin{align*}
\left| \left[ \eta_i + \tau_i(z^{-\xi}-1)/\xi \right]^{\alpha} \exp(-z) \right| &= \left[ \eta_i + \tau_i(z^{-\xi}-1)/\xi \right]^{\alpha} \exp(-z) \\& \leq \left[ \eta_i + \tau_i(z^{-\xi^*}-1)/\xi^* \right]^{\alpha} \exp(-z)
\end{align*}
and therefore, taking $\alpha=\beta_i (1 + \varepsilon)$,
\begin{align*}
& \quad \ \sup_{\xi \in (0, \xi^*]} \int_{0}^{1} \left[ \eta_i + \tau_i(z^{-\xi}-1)/\xi \right]^{\beta_i (1 + \varepsilon)} \exp(-z) \mathrm{d}z \\&= \int_{0}^{1} \left[ \eta_i + \tau_i(z^{-\xi^*}-1)/\xi^* \right]^{\beta_i (1 + \varepsilon)} \exp(-z) \mathrm{d}z < \infty.
\end{align*}
Combining this result with a similar reasoning for the first integral in \eqref{Eq_ExpressionMomentalpha} and using \eqref{Eq_Inequality_SupMomentalpha} yields $\sup_{\xi \in (0, K]}\mathbb{E}[ |X_{i, \xi}^{\beta_i}|^{1+\varepsilon}]< \infty$ for some $K>0$. Now, let $Y_{\xi}=X_{1, \xi}^{\beta_1} X_{2, \xi}^{\beta_2}$ and $Y_{0}=X_{1, 0}^{\beta_1} X_{2, 0}^{\beta_2}$. 
By Cauchy--Schwarz inequality,
$$
\sup_{\xi \in (0, K]} \mathbb{E}\left[ \left| Y_{\xi} \right|^{1+\varepsilon} \right] \leq \sqrt{\sup_{\xi \in (0, K]} \mathbb{E} \left[ \left| X_{1, \xi}^{\beta_1} \right|^{2(1+\varepsilon)}\right]} \sqrt{\sup_{\xi \in (0, K]} \mathbb{E} \left[ \left| X_{2, \xi}^{\beta_2} \right|^{2(1+\varepsilon)}\right]} < \infty.
$$
It follows from \citet[][p.31]{billingsley1999convergence} that the $(X_{1, \xi})_{\xi}$, $(X_{2, \xi})_{\xi}$ and $(Y_{\xi})_{\xi}$ are uniformly integrable for $\xi$ around $0$ (from the right).

Now, it is well-known that $X_{i, \xi} \overset{d}{\to} X_{i, 0}$, $i=1, 2$, which implies by the continuous mapping theorem that $X_{i, \xi}^{\beta_i} \overset{d}{\to} X_{i, 0}^{\beta_i}$. Moreover, for any $z_1, z_2 \in \mathbb{R}$,
\begin{align*}
\mathbb{P} \left( \left[Z_1^{\xi}-1\right]/\xi \leq z_1, \left[Z_2^{\xi}-1 \right]/\xi \leq z_2 \right) &= \mathbb{P} \left( Z_1 \leq (1+\xi z_1)^{1/\xi}, Z_2 \leq (1+\xi z_2)^{1/\xi} \right) \\
&= \exp \left( -V \left([1+\xi z_1]^{1/\xi}, [1+\xi z_2]^{1/\xi} \right) \right),
\end{align*}
and 
$$ \mathbb{P} \left( \log Z_1 \leq z_1, \log Z_2 \leq z_2 \right) = \exp(-V(\exp(z_1), \exp(z_2))),$$
where $V$ is the exponent function of $(Z_1, Z_2)'$.
Thus, by continuity of $V$,
$$ \lim_{\xi \to 0} \mathbb{P} \left( \left[Z_1^{\xi}-1\right]/\xi \leq z_1, \left[Z_2^{\xi}-1\right]/\xi \leq z_2 \right)= 
\mathbb{P} \left( \log Z_1 \leq z_1, \log Z_2 \leq z_2 \right),$$
and therefore $$\left( \left[Z_1^{\xi}-1\right]/\xi, \left[Z_2^{\xi}-1\right]/\xi \right)' \overset{d}{\to} \left( \log Z_1, \log Z_2 \right)'.$$
Consequently, the continuous mapping theorem yields 
$$(X_{1, \xi}^{\beta_1}, X_{2, \xi}^{\beta_2})' \overset{d}{\to} (X_{1, 0}^{\beta_1}, X_{2, 0}^{\beta_2})',$$
and hence, applied again, $Y_{\xi} \overset{d}{\to} Y_{0}$.
Finally, Theorem 3.5 in \cite{billingsley1999convergence} yields that
$\lim_{\xi \to 0} \mathbb{E}(X^{\beta_i}_{i, \xi}) = \mathbb{E}(X^{\beta_i}_{i, 0})$, $i=1,2$ and $\lim_{\xi \to 0} \mathbb{E}(Y_{\xi})=\mathbb{E}(Y_{0})$. The result follows immediately. Similar arguments give the same conclusion for $\xi<0$.
\end{proof}

\subsection{For Proposition \ref{Prop_Decrease_gtilde}}

In this section, we denote by $F_{X}$ the distribution function of any random variable $X$ and by $F_{X_1, X_2}$ the distribution function of any random vector $\bm{X}=(X_1, X_2)'$.

\subsubsection{Preliminary result}
\label{subsubsec_Generalization_Dhaene}

We first need the following result.
\begin{Prop}
\label{Prop_Generalization_Dhaene}
Let $\bm{X}=(X_1, X_2)'$ and $\bm{Y}=(Y_1, Y_2)'$ be random vectors such that $F_{X_1}=F_{Y_1}$ and $F_{X_2}=F_{Y_2}$. We have
$$ 
F_{X_1, X_2}(z_1, z_2) < F_{Y_1, Y_2}(z_1, z_2) \ \mbox{ for all } z_1, z_2 > 0 \Longrightarrow  \Tb{\mathrm{Cov}(f_1(X_1), f_2(X_2)) < \mathrm{Cov}(f_1(Y_1), f_2(Y_2))},
$$
for all strictly increasing functions $f_1: \Tb{(0, \infty)} \to \Tb{\mathbb{R}}$ and $f_2: \Tb{(0, \infty)} \to \Tb{\mathbb{R}}$, provided the covariances exist.
\end{Prop}
\begin{proof}
The proof is partly inspired from the proof of Theorem 1 in \cite{dhaene1996dependency}. Let $f_1: \Tb{(0, \infty)} \to \Tb{\mathbb{R}}$ and $f_2: \Tb{(0, \infty)} \to \Tb{\mathbb{R}}$ be strictly increasing functions. Assume that, for all $z_1, z_2>0$, 
\Beq
\label{Eq_Majoration_Function_Distribution}
F_{X_1, X_2}(z_1, z_2) < F_{Y_1, Y_2}(z_1, z_2).
\Eeq
We have 
$$ \mathbb{P}(f_1(X_1) \leq z_1, f_2(X_2) \leq z_2)= \mathbb{P} \left( X_1 \leq f_1^{-1}(z_1), X_2 \leq f_2^{-1}(z_2) \right)$$
and the same equality for $\bm{Y}$. Consequently, since, for all $z_1, z_2>0$, $f_1^{-1}(z_1), f_2^{-1}(z_2)>0$, it follows from 
\eqref{Eq_Majoration_Function_Distribution} that, for all $z_1, z_2>0$,
\Beq
\label{Eq_Majoration_Function_Distribution_Transformed}
\mathbb{P}(f_1(X_1) \leq z_1, f_2(X_2) \leq z_2) < \mathbb{P}(f_1(Y_1) \leq z_1, f_2(Y_2) \leq z_2).
\Eeq
Since $X_1$ and $Y_1$ have the same distribution and this also holds for $X_2$ and $Y_2$, we deduce that
\Beq
\label{Eq_Same_Margins}
f_1(X_1) \overset{d}{=} f_1(Y_1) \quad \mbox{and} \quad f_2(X_2) \overset{d}{=} f_2(Y_2).
\Eeq
Using \eqref{Eq_Majoration_Function_Distribution_Transformed}, \eqref{Eq_Same_Margins} and  Lemma 1 in \cite{dhaene1996dependency}, we obtain
\begin{align*}
\mathrm{Cov}(f_1(X_1), f_2(X_2)) &=\int_{0}^{\infty} \int_{0}^{\infty} \left[ F_{f_1(X_1), f_2(X_2)}(u,v)-F_{f_1(X_1)}(u) F_{f_2(X_2)}(v)\right] \mathrm{d}u \mathrm{d}v \\ 
& < \int_{0}^{\infty} \int_{0}^{\infty} \left[ F_{f_1(Y_1), f_2(Y_2)}(u,v)-F_{f_1(Y_1)}(u) F_{f_2(Y_2)}(v) \right] \mathrm{d}u \mathrm{d}v \\
&= \mathrm{Cov}(f_1(Y_1), f_2(Y_2)).
\end{align*}
\end{proof}

\subsubsection{Proof of Proposition \ref{Prop_Decrease_gtilde}}

\begin{proof}
Let $\bm{Z}=(Z_1, Z_2)'$ be a random vector having the Hüsler--Reiss distribution function \eqref{Eq_HuslerReissDist} with parameter $h$. We immediately obtain that, for all $z_1, z_2>0$,
\Beq
\label{Eq_Derivative_h_Bivariate_Df_Smith}
\frac{\partial \mathbb{P}(Z_1 \leq z_1, Z_2 \leq z_2)}{\partial h}(h) = \exp \left( -\frac{1}{z_1} \Phi\left( \frac{h}{2} + \frac{1}{h} \log \left( \frac{z_2}{z_1} \right) \right) -\frac{1}{z_2} \Phi\left( \frac{h}{2} + \frac{1}{h} \log \left( \frac{z_1}{z_2} \right) \right) \right) T_2,
\Eeq
where
$$ T_2 = -\frac{1}{z_1} \left(\frac{1}{2} - \frac{\log(z_2/z_1)}{h^2} \right) \phi \left( \frac{h}{2}+\frac{\log(z_2/z_1)}{h} \right) - \frac{1}{z_2} \left(\frac{1}{2} + \frac{\log(z_2/z_1)}{h^2} \right) \phi \left( \frac{h}{2}-\frac{\log(z_2/z_1)}{h} \right).
$$
For all $z_1, z_2>0$, we introduce $y=z_2/z_1$, which is strictly positive. We have
\begin{align*}
T_2 &= \frac{1}{z_2} \left[-\frac{z_2}{z_1} \left(\frac{1}{2} - \frac{\log(z_2/z_1)}{h^2} \right) \phi \left( \frac{h}{2}+\frac{\log(z_2/z_1)}{h} \right) - \left(\frac{1}{2} + \frac{\log(z_2/z_1)}{h^2} \right) \phi \left( \frac{h}{2}-\frac{\log(z_2/z_1)}{h} \right) \right]
\\& = \frac{1}{z_2} \left[-y \left(\frac{1}{2} - \frac{\log y}{h^2} \right) \phi \left( \frac{h}{2}+\frac{\log y}{h} \right) - \left(\frac{1}{2} + \frac{\log y}{h^2} \right) \phi \left( \frac{h}{2}-\frac{\log y}{h} \right) \right]
\\& = \frac{1}{\sqrt{2 \pi} z_2} \exp\left( -\frac{h^2}{8}-\frac{(\log y)^2}{2h^2} \right) \left[ -y \left(\frac{1}{2} - \frac{\log y}{h^2} \right)y^{-1/2} - \left(\frac{1}{2} + \frac{\log y}{h^2} \right) y^{1/2} \right]
\\& = -\frac{y^{1/2}}{\sqrt{2 \pi} z_2} \exp\left( -\frac{h^2}{8}-\frac{(\log y)^2}{2h^2} \right),
\end{align*}
which is strictly negative. Thus, \eqref{Eq_Derivative_h_Bivariate_Df_Smith} gives that, for all $h \geq 0$ and $z_1, z_2>0$, 
\Beq
\label{Eq_Derivative_h_Bivariate_Df_Smith_Negative}
\partial \mathbb{P}(Z_1 \leq z_1, Z_2 \leq z_2)/\partial h(h) <0.
\Eeq
Let us consider $h_1>h_2>0$, and $\bm{Z}_{1}=(Z_{1, 1}, Z_{1, 2})'$ and $\bm{Z}_{2}=(Z_{2, 1}, Z_{2, 2})'$ following the Hüsler--Reiss distribution \eqref{Eq_HuslerReissDist} with parameters $h_1$ and $h_2$, respectively.
We get from \eqref{Eq_Derivative_h_Bivariate_Df_Smith_Negative} that
$F_{Z_{1, 1}, Z_{1, 2}}(z_1, z_2) < F_{Z_{2, 1}, Z_{2, 2}}(z_1, z_2)$ for all $z_1, z_2 >0$. 
Since the components of $\bm{Z}_{1}$ and $\bm{Z}_{2}$ all follow the standard Fréchet distribution, we have $F_{Z_{1, 1}}=F_{Z_{2, 1}}$ and $F_{Z_{1, 2}}=F_{Z_{2, 2}}$. 
Now, as $\tau>0$, for $\xi \neq 0$, the function 
$$
\begin{array}{cccc}
f: & (0, \infty) & \to & \mathbb{R} \\
 & z & \mapsto & \left( \eta-\tau/\xi + \tau z^{\xi}/\xi \right)^{\beta}
\end{array}
$$
is strictly increasing. 
Hence, letting 
$$ Y_{i, j} = \eta-\frac{\tau}{\xi} + \frac{\tau}{\xi} {Z_{i, j}}^{\xi}, \quad i, j=1, 2,$$
Proposition \ref{Prop_Generalization_Dhaene} yields
\Beq
\label{Eq_Order_Expectation}
\mathrm{Cov} \left( Y_{1, 1}^{\beta}, Y_{1, 2}^{\beta} \right) < \mathrm{Cov} \left( Y_{2, 1}^{\beta}, Y_{2, 2}^{\beta} \right).
\Eeq
Furthermore, we know from \eqref{Eq_Cov_Maxstab_Real_Marg_Eq_Coeff} that, for $i=1, 2$, 
\Beq
\label{Eq_g_Expectation}
\mathrm{Cov} \left( Y_{i, 1}^{\beta}, Y_{i, 2}^{\beta} \right)=g_{\beta, \eta, \tau, \xi}(h_i)-\sum_{k_1=0}^{\beta} \sum_{k_2=0}^{\beta} B_{k_1, k_2, \beta, \eta, \tau, \xi} \ \Gamma(1-[\beta-k_1]\xi) \Gamma(1-[\beta-k_2]\xi).
\Eeq
Finally the combination of \eqref{Eq_Order_Expectation} and \eqref{Eq_g_Expectation} gives that
$g_{\beta, \eta, \tau, \xi}(h_1)<g_{\beta, \eta, \tau, \xi}(h_2)$, showing the result. 
\end{proof}

\subsection{For Proposition \ref{Prop_Lim_gtildebeta_hto0}}

\begin{proof}
Let $X$ be the Brown--Resnick field associated with the semivariogram $\gamma_W(\bm{x})=\| \bm{x}\|^2/2$, $\bm{x} \in \Mbb{R}^2$, and with GEV parameters $\eta$, $\tau$, and $\xi \neq 0$, and $\beta \in \mathbb{N}_*$ such that $\beta \xi < 1/2$. It is well-known that $X$ is sample-continuous.

The field $X^{\beta}$ is stationary by stationarity of $X$ and has a finite second moment since $\beta \xi < 1/2$. Accordingly, $X^{\beta}$ is second-order stationary. Moreover, $X^{\beta}$ is sample-continuous and thus, by the same arguments as in the proof of Proposition 1 in \cite{koch2017TCL}, continuous in quadratic mean. Hence, the covariance function of $X^{\beta}$ is continuous at the origin. It implies by Theorem \ref{ThAppBR} that
\begin{align*}
& \quad \lim_{\bm{x} \to \bm{0}} \mathrm{Cov} \left( X^{\beta}(\bm{0}), X^{\beta}(\bm{x}) \right) \\&=\lim_{\bm{x} \to \bm{0}} \left( g_{\beta, \eta, \tau, \xi} \left( \| \bm{x} \| \right)- \sum_{k_1=0}^{\beta} \sum_{k_2=0}^{\beta} B_{k_1, k_2, \beta, \eta, \tau, \xi} \ \Gamma(1-[\beta-k_1]\xi) \Gamma(1-[\beta-k_2]\xi) \right) \nonumber
\\& = \Mr{Var} \left( X^{\beta}(\bm{0}) \right), 
\end{align*}
which, combined with \eqref{Eq_Var_GEV_beta}, yields \eqref{Eq_Lim_gtildebeta_hto0}.
This easily gives $\lim_{h \to 0}  g_{\beta, \eta, \tau, \xi}(h)= g_{\beta, \eta, \tau, \xi}(0)$, which implies that $g_{\beta, \eta, \tau, \xi}$ is continuous at $h=0$. The continuity of $g_{\beta, \eta, \tau, \xi}$ at any $h>0$ comes from the fact that the covariance function of a field which is second-order stationary can be discontinuous only at the origin.
\end{proof}

\subsection{For Proposition \ref{Prop_Lim_gtilde_infty}}

\subsubsection{Preliminary results}

\begin{Lem}
\label{Lemma_CP_Locally_Integrable_Sample_Paths}
Let $\{ X(\bm{x}) \}_{\bm{x} \in \mathbb{R}^2}$ be a measurable max-stable random field with GEV parameters $\eta \in \mathbb{R}$, $ \tau>0$ and $\xi \neq 0$. Let $\beta \in \mathbb{N}_*$ such that $\beta \xi < 1$. Then, the random field $X^{\beta}$ belongs to $\mathcal{C}$.
\end{Lem}
\begin{proof}
The field $X^{\beta}$ is obviously measurable. Furthermore, as $X$ has identical univariate marginal distributions,
the function $\bm{x} \mapsto \mathbb{E} [ |X(\bm{x})^{\beta} | ]$
is constant and hence locally integrable. Therefore, Proposition 1 in \cite{koch2019SpatialRiskAxioms} yields that $X^{\beta}$ has a.s. locally integrable sample paths.
\end{proof}

Let $\mathcal{B}(\mathbb{R})$ and $\mathcal{B}((0, \infty))$ denote the Borel $\sigma$-fields on $\mathbb{R}$ and $(0, \infty)$, respectively. 
\begin{Lem}
\label{Lem_Moment_2plusdelta}
Let $\{ Z(\bm{x}) \}_{\bm{x} \in \mathbb{R}^2}$ be a simple max-stable random field. Let $\eta \in \mathbb{R}$, $\tau>0$, $\xi \in \mathbb{R}$ and $\beta \in \mathbb{N}_*$.
The function defined by
\Beq
\label{Eq_D_beta_eta_tau_xi}
D_{\beta, \eta, \tau, \xi}(z)=
\left\{
\begin{array}{ll}
\left( \eta-\tau/\xi + \tau z^{\xi}/\xi \right)^{\beta}, & \quad \xi \neq 0, \\
\left( \eta + \tau \log z \right)^{\beta}, & \quad \xi = 0,
\end{array}
\qquad z >0, \right.
\Eeq
is measurable from $((0, \infty), \mathcal{B}((0,\infty)))$ to $(\mathbb{R},\mathcal{B}(\mathbb{R}))$ and strictly increasing. Moreover, if $\beta \xi <1/2$, then $\mathbb{E} [ | D_{\beta, \eta, \tau, \xi}(Z(\bm{0})) |^{2+\delta} ]<~\infty$ for any $\delta$ such that $0 < \delta < 1/(\xi \beta)-2$.
\end{Lem}
\begin{proof}
The fact that $D$ is measurable and strictly increasing is obvious. Denoting $Y= [D_{\beta, \eta, \tau, \xi}(Z(\bm{0}))]^{1/\beta}$, we have, for $\delta>0$,
$$\mathbb{E} \left[ \left| D_{\beta, \eta, \tau, \xi}(Z(\bm{0})) \right|^{2+\delta} \right]=\mathbb{E} \left[ \left|Y^{\beta} \right|^{2+\delta}\right] = \mathbb{E}\left[\left |Y \right|^{\beta(2+\delta)}\right],$$
which is finite (see the proof of Proposition \ref{Prop_Continuity_Cov_xi_0}) provided $\beta(2+\delta) \xi < 1$ as $Y$ follows the GEV distribution with parameters $\eta$, $\tau$ and $\xi$. The latter inequality is satisfied for any strictly positive $\delta$ such that $\delta < 1/(\xi \beta)-2$. 
\end{proof}

\subsubsection{Proof of Proposition \ref{Prop_Lim_gtilde_infty}}

\begin{proof}
Let $X$ be the Brown--Resnick field associated with the semivariogram $\gamma_W(\bm{x})=\| \bm{x}\|^2/2$, $\bm{x} \in \Mbb{R}^2$, and with GEV parameters $\eta$, $\tau$ and $\xi \neq 0$, and $\beta \in \mathbb{N}_*$ such that $\beta \xi < 1/2$. 

The field $X$ is sample-continuous 
and thus measurable, which yields by Lemma \ref{Lemma_CP_Locally_Integrable_Sample_Paths} that $X^{\beta} \in \mathcal{C}$. Now, we have $X^{\beta}(\bm{x})=D_{\beta, \eta, \tau, \xi}(Z(\bm{x}))$, $\bm{x} \in \Mbb{R}^2$, where $Z$ is the simple Brown--Resnick field associated with the semivariogram just above, and $D_{\beta, \eta, \tau, \xi}$ is defined in \eqref{Eq_D_beta_eta_tau_xi}. In addition, by Lemma \ref{Lem_Moment_2plusdelta}, $D_{\beta, \eta, \tau, \xi}$ satisfies the assumptions on the function $F$ of Theorem 3 in \cite{koch2017TCL}. Thus, the latter theorem yields that $X^{\beta}$ satisfies the central limit theorem. This implies that 
$$
\int_{\mathbb{R}^2} \left| \mathrm{Cov} \left( X^{\beta}(\bm{0}), X^{\beta}(\bm{x}) \right) \right|\mathrm{d}\bm{x} < \infty,
$$
which entails, using Theorem \ref{ThAppBR}, that
$$ \int_{\mathbb{R}^2} \left( g_{\beta, \eta, \tau, \xi} \left( \| \bm{x} \| \right) - \sum_{k_1=0}^{\beta} \sum_{k_2=0}^{\beta} B_{k_1, k_2, \beta, \eta, \tau, \xi} \ \Gamma(1-[\beta-k_1]\xi) \Gamma(1-[\beta-k_2]\xi) \right) \mathrm{d}\bm{x} < \infty.$$
Since $g_{\beta, \eta, \tau, \xi}$ is strictly decreasing, this necessarily implies that 
$$ \lim_{h \to \infty} \left( g_{\beta, \eta, \tau, \xi} \left( h \right) - \sum_{k_1=0}^{\beta} \sum_{k_2=0}^{\beta} B_{k_1, k_2, \beta, \eta, \tau, \xi} \ \Gamma(1-[\beta-k_1]\xi) \Gamma(1-[\beta-k_2]\xi) \right)=0,$$
i.e., \eqref{Eq_Lim_gtilde_infty}. \end{proof}

\section{Case of simple Brown--Resnick fields and $\beta<1/2$}
\label{Sec_Appendix_SimpleBRfield}

This appendix explains that the results obtained in Sections~\ref{Subsec_TheoreticalContribution} and \ref{Subsec_Results} are similar if the Brown--Resnick field considered is simple and the power satisfies $\beta<1/2$. 
As standard Fr\'echet margins are rarely encountered in practice, the interest of this section mostly lies in a better understanding of some properties of simple Brown--Resnick fields and in possible applications to inference (using, e.g., the method of moments).

First we consider the dependence measure $\mathrm{Corr} ( Z^{\beta}(\Mb{x}_1), Z^{\beta}(\Mb{x}_2))$, where $\{ Z(\Mb{x}) \}_{\Mb{x} \in \mathbb{R}^2}$ is a simple Brown--Resnick max-stable random field and $\beta < 1/2$. The condition $\beta \xi < 1/2$ with $\beta \in \mathbb{N}_*$ of \eqref{Eq_DepMeas} translates into $\beta < 1/2$; any negative value is allowed as simple max-stable fields are a.s. strictly positive. We introduce, for $\beta<1/2$,  
$$
I_{\beta}(h) =
\left \{
\begin{array}{ll}
\Gamma(1-2 \beta) & \mbox{if} \quad  h=0, \\ 
\displaystyle \int_{0}^{\infty} \theta^{\beta} \Big[ C_2(\theta,h) \  C_1(\theta,h)^{2 \beta -2} \ \Gamma(2-2 \beta) 
\\ \qquad + C_3(\theta,h) \ C_1(\theta,h)^{2 \beta -1} \  \Gamma(1-2 \beta) \Big] \  \mathrm{d}\theta & \mbox{if} \quad h>0,
\end{array}
\right.
$$
which arises when setting $\beta_1=\beta_2$ in the function $I_{\beta_1, \beta_2}$ specified in \eqref{Eq_Def_g_beta1_beta2}.
Denoting by $\gamma_W$ the semivariogram of $Z$, it follows from Theorem \ref{Th_CovarianceSimpleHR} and \eqref{EqBivDistFuncBRField} that, for all $\Mb{x}_1, \Mb{x}_2 \in \mathbb{R}^2$ and $\beta < 1/2$,
$\mathrm{Cov} ( Z^{\beta}(\Mb{x}_1), Z^{\beta}(\Mb{x}_2) )=I_{\beta}(\sqrt{2 \gamma_W(\Mb{x}_2-\Mb{x}_1)}) - \left[ \Gamma(1- \beta) \right]^2$. Then $\mathrm{Corr} ( Z^{\beta}(\Mb{x}_1), Z^{\beta}(\Mb{x}_2))$ (provided that $\beta \neq 0$) is readily derived and its behaviour is similar to the one we observed in Section \ref{Subsec_Results} (not shown); for more details, see Figures 3 and 4 in the unpublished work by \cite{koch2019spatialpowers1}.

We now investigate the function $I_{\beta}$ in further details.
Very similar proofs as for Propositions~\ref{Prop_Decrease_gtilde}--\ref{Prop_Lim_gtilde_infty} yield, for $\beta, \beta_1, \beta_2 <1/2$, that the functions $I_{\beta_1, \beta_2}$ defined in \eqref{Eq_Def_g_beta1_beta2} and $I_{\beta}$ are strictly decreasing, $\lim_{h \to 0} I_{\beta}(h)=\Gamma(1-2 \beta)$
(implying that $I_{\beta}$ is continuous everywhere on $[0, \infty)$) and $\lim_{h \to \infty} I_{\beta}(h) = [\Gamma(1-\beta)]^2$. This entails that, for any $h \geq 0$, $\lim_{\beta \to -\infty} I_{\beta}(h)= \infty$.
Figure \ref{Plot_Persp_Ibeta}, obtained using adaptive quadrature with a relative accuracy of $10^{-5}$, shows that the decrease of $I_{\beta}(h)$ for a given $\beta$ with respect to $h$ is more and more pronounced when $|\beta|$ increases, and that, for $h$ fixed, the absolute value of the slope of $I_{\beta}(h)$ increases very fast with $|\beta|$, in link with rapid divergence to $\infty$. Obviously, the behaviour of $\mathrm{Cov} ( Z^{\beta}(\Mb{x}_1), Z^{\beta}(\Mb{x}_2))$ is similar; the same holds true for $I_{\beta_1, \beta_2}$.
\begin{figure}[h!]
\center
\includegraphics[scale=0.8]{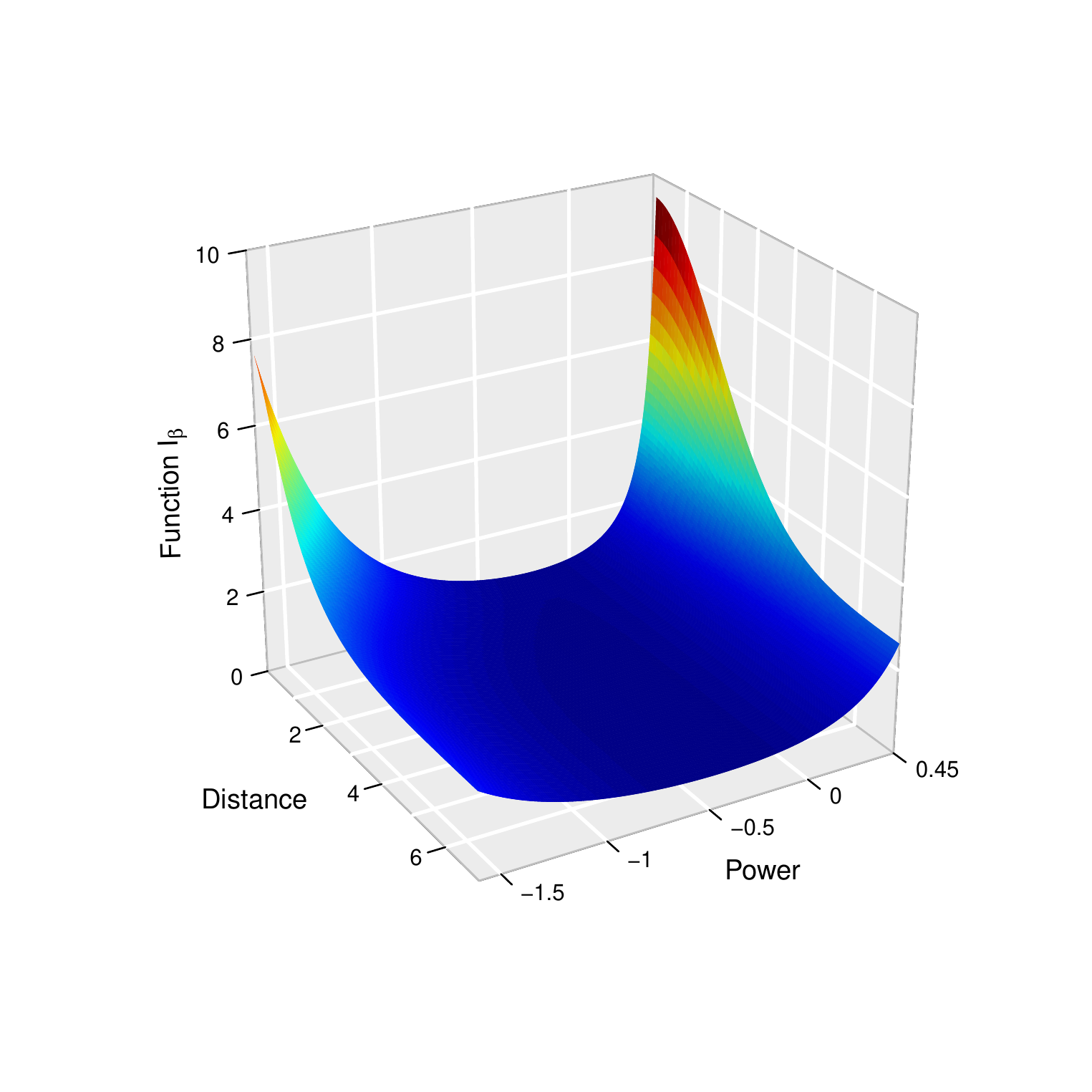}
\caption{Evolution of the function $I_{\beta}$ with respect to the distance $h$ and the power $\beta$ for $\beta \in [-1.6, 0.45]$.}
\label{Plot_Persp_Ibeta}
\end{figure}

\newpage
\bibliographystyle{apalike}
\bibliography{../../../../../Bibliography_Erwan/References_Erwan}

\end{document}